\documentclass[11pt]{article}
\usepackage[T1]{fontenc}
\usepackage{textcomp}
\normalfont

\usepackage{graphicx}
\usepackage{enumerate}
\usepackage{amsmath}
\usepackage{amsfonts}
\usepackage{amssymb}
\usepackage{amsfonts}
\usepackage{hyperref}
\usepackage{amsthm,epsfig}
\usepackage{bm}
\usepackage[margin=1in]{geometry}
\usepackage{color,xcolor,tikz}
\usepackage{enumerate}
\usepackage{xcolor}
\usepackage{accents}
\usepackage{natbib}  
\usepackage{subfig}

\newcommand {\be}{\begin{equation}}
\newcommand {\ee}{\end{equation}}
\newcommand{\payoff}[1]{\mathit{\Pi}_{#1}}
\newcommand{\Gammait}{\mathit{\Gamma}}

\definecolor{redextra}{HTML}{D95A4D}
\definecolor{greenextra}{HTML}{2F714B}

\hypersetup{colorlinks=true,linkcolor=greenextra,urlcolor=purple,citecolor=greenextra} 

\newtheorem{theorem}{Theorem}
\newtheorem{proposition}{Proposition}
\newtheorem{lemma}{Lemma}

\newtheorem*{definition*}{Definition}
\newtheorem{example}{Example}
\newtheorem*{example*}{Example}

\newtheorem*{remark*}{Remark}
\newtheorem*{observation*}{Observation}

\newcommand{\e}{\bm{e}}

\usepackage{soul}

\begin{document}

\title{Interconnected Contests\footnote{We thank Chengqing Li and Yiyi Li for excellent research assistance. We also thank Francis Bloch, Matthew Elliott, John Geanakoplos, Dan Kovenock, Larry Samuelson, Alex Wolitzky, and participants at a number of conferences and seminars for useful comments. Dziubi\'{n}ski's work was supported by the Polish National Science Centre through grants 2018/29/B/ST6/00174 and 2021/42/E/HS4/00196. Goyal gratefully acknowledges financial support from Janeway Institute and Keynes Fund at Cambridge University University.  Zhou gratefully acknowledges financial support from the NSFC (Grant No. 72450001). }}

\author{Marcin Dziubi\'{n}ski\footnote{Institute of Informatics, University of Warsaw. \texttt{m.dziubinski@mimuw.edu.pl}}
\and Sanjeev Goyal\footnote{University of Cambridge.   \texttt{sg472@cam.ac.uk}}
\and Junjie Zhou\footnote{School of Economics and Management, Tsinghua University. \texttt{zhoujj03001@gmail.com}}
}

\date{\today}

\maketitle

\begin{abstract}
We study a two-player model of conflict with multiple battlefields -- the novel element is that each of the players has their own network of spillovers so that resources allocated to one battle can be utilized in winning neighboring battles. There exists a unique equilibrium in which the relative probability of a player winning a battle is the product of the ratio of the centrality of the battlefield in the two respective competing networks and the ratio of the relative cost of efforts of the two players. We study the design of networks and characterize networks that maximize total efforts and maximize total utility. Finally, we characterize the equilibrium of a game in which players choose both networks and efforts in the battles.   

\end{abstract}

\newpage
\section{Introduction}

 Following \cite{Borel1921}, we consider two players that allocate resources across battlefields so as to win the maximum number of them. The novel element in our work is that each of the players has their own network of spillovers: resources allocated by a player to one battle can be utilized in part by that player in contests on connected/neighbouring battles. We study three broad sets of questions. First, we take the competing networks as given and we ask how they  affect the efforts and the payoffs of the players. Second, we turn to the issue of designing networks and we ask what is the architecture of networks that maximizes aggregate efforts and maximizes player utility. Finally, we decentralize the choice of networks and we ask what is the architecture of networks and the configuration of efforts that arise when players choose both the spillovers and the efforts.\footnote{Applications of our model range widely and they include military conflicts \citep{garfinkel2012oxford,KovenockRoberson2018}, political competition \citep{brams19743,snyder1989election}, advertising competition \citep{Friedman1958}, research and innovation \cite{dasgupta1980industrial},\citet*{BayeKovenockdeVries1984}) and network security (\cite{AlpcanBasar11}, \citet*{Cunningham85}, \cite{DziubinskiGoyal13}, \citet*{GoyalVigier14}; for a survey see \cite{Goyal2023}). \newline 
Spillovers arise naturally in these applications. In advertising competition investments in advertising on one media channel influence other media channels. In political campaigns, expenditures in one constituency may have a bearing on `neighboring' constituencies. In military applications, army assignments to one battlefield may be moved to neighboring  battlefields, depending on the structure of the infrastructure. The issue of network choice is quite natural -- anticipating the important role of transport, generals and governments make decisions on investments in infrastructure and these decisions have a bearing on their allocations of armies and other resources across different battlefields. } 

We begin by noting that, in the absence of networks, the theory of contests tells us a player will allocate resources to a battlefield in inverse proportion to its relative costs and in proportion to the value of the prize. When efforts on a battlefield have spillovers, the marginal returns to efforts equals the marginal returns on that battlefield and on the linked battlefields: these marginal returns will depend on allocations of the player on different battlefields and on the spillovers across the battlefields but they will also depend on the efforts of the other player, which in turn will depend on the network of that player. This interaction between spillovers on own network and on the competing network is central to our analysis. 

Suppose for simplicity that there are only two battlefields. Keeping other parts of the setting as fixed, as we increase the spillover from battlefield A to battlefield B the marginal returns of allocations to battlefield A will increase. As spillovers become large, at some point it seems reasonable to expect that a player will only allocate resources on battlefield A. This suggests that corner solutions will be a natural feature of our setting, a fact that greatly enriches but also complicates the analysis.  

We begin the general analysis by showing there exists a unique (Nash) equilibrium in this game -- this may be an interior or a corner equilibrium. The equilibrium total efforts of the players and their equilibrium payoffs depend on the relative probabilities of winning, which in turn depend on the costs of players, the networks of spillovers, and the values of the battlefields. We provide a closed form solution of these relationships. 

In particular, we establish that the relative probability of player 1 winning a battlefield is equal to the product of the ratio of relative costs and the ratio of Bonacich centrality of this node in the two competing networks. This result brings together the theory of contests with a key insight of the theory of games on networks.\footnote{We note that these properties hold both in the interior and in the corner equilibrium.}

An implication of the result on ratios of costs and network centrality is the following: if the two competing networks are the same, then the centralities are equal and they cancel out. This means that the relative probability of a player winning a battle will depend only on the relative costs of effort. We noted above that the equilibrium total efforts of the players depend on these probabilities. So this means that the aggregate efforts and hence the payoffs are also independent of the networks of spillovers. In other words, the probabilities of winning the battlefields, the total efforts, and the payoffs of the players are independent of the network of spillovers, if the two competing networks are equal. So, for the networks to matter, the two networks have to be different. 

Equipped with this result, we turn to the design of networks that maximize total efforts and total utility. Our equilibrium characterization provides a formula for total efforts of each player. This formula tells us that the player's aggregate efforts are the sum of product of probabilities of the two players winning each of the nodes. As probabilities sum to 1, the product is maximized when both players have an equal probability of winning every node. Given profile of cost of effort, this equality is obtained when the network spillover compensates the higher-cost player appropriately. For equal value battles, this intuition immediately yields our next result: the network that maximizes the sum total of aggregate efforts involves an empty network for the cost advantaged player and the creation of links with strength that offers a \textit{compensating handicap} to the cost disadvantaged player. This idea can be extended to cover unequal value of battles with an appropriate change in networks. We also study the design of networks that maximize utility. When the values
of the battlefields are equal, an empty network for the higher-cost player and a
complete network of spillovers with sufficiently large magnitude of spillovers for
the other player supports very small efforts from both player. This in turn allows equilibrium payoffs to be close to the first-best payoffs.  

The final part of the paper considers a model in which players choose efforts on battles and also choose the strength of spillovers across battlefields.  As costs of links are zero, to make the problem interesting, we assume that  link strength can take value between zero and one. We show that, because costs of links are zero, in equilibrium, the efforts of a player at any battlefield are available at every battle field. This \textit{universal access} property implies that players choose aggregate efforts that are equal to the effort they would choose in a single contest whose prize is equal in value to the sum of the value of prizes at all battlefields. Hence, equilibrium efforts depend only on the relative costs and the value of prizes. This outcome is very different from the networks that maximize aggregate efforts or maximize utility. 

The study of conflict on multiple battlefields is one of the oldest problems in game theory (\cite{Borel1921}). We consider a smooth contest on battlefields and we extend the multi-battlefield model to allow for spillovers in effort assignments across the battlefields. Our paper may be seen as a bridge between the theory of contests and the theory of network games.\footnote{Contributions to the theory of contests include \cite{Borel1921}, \cite{dixit1987strategic},\cite{Friedman1958}, \cite{fu2020optimal}, \cite{konrad2009multi}, and \cite{KovenockRoberson2012}; for surveys of this literature, see \cite{chowdhuryetal2025} and \cite{Konrad09}. Contributions to the theory of network games include \citet*{ballester2006whoiswho}, \citet*{bramoulle2014strategic} and \cite{goyal2001r}. For contributions to the theory of contests in networks, see \cite{FrankeOzturk2015}, \citet*{konig2017networks} and \citet*{XuZenouZhou2022}. In these papers the nodes in the network are the players who choose efforts in contests; by contrast, in our paper, the players are located outside the network and they allocate resources on all nodes of the network. We see our model as covering a different set of applications -- those identified in the first footnote above.} In particular, our characterization of equilibrium in terms of relative costs and relative centralities of nodes across two networks is novel in the context of the literature on competition on networks (\cite{bloch2013pricing}, \citet*{candogan2012optimal}, \cite{fainmesser2016pricing,fainmesser2020pricing}, \citet*{chen2018competitive} and \cite{GoyalKearns2012}. For a recent survey of the literature, see \cite{Goyal2023}). Our results on the design of networks that maximize aggregate efforts (and aggregate utility) appear to be novel in the context of the literature on contest design  (see e.g., \cite{fang2020turning}, \cite{fu2020optimal}, \cite{halac2017contests},   \cite{hinnosaar2024optimal}, \cite{Moldovanu2001optimal}, \cite{moldovanu2007contests}, \cite{olszewski2020performance}).\footnote{The result on using networks with greater spillovers to compensate the higher cost player is reminiscent of the idea of handicapping that has been discussed in the contest literature; see e.g.,\cite{che2003optimal}.} Finally, our result on endogenous networks appears to be novel in the context of the literature on endogenous networks (\cite{BalaGoyal00a}, \cite{galeotti2010law}, \citet*{HendricksPiccioneTan1999}, \cite{Huremovic2021}, \cite{vegaredondo2016}, and \cite{Goyal2023}.)

The rest of the paper is organized as follows. In Section~\ref{sec:model} we define the model.  In Section~\ref{sec:eqprop} we establish existence and uniqueness and characterize the equilibrium.  Section~\ref{sec:spilldes} studies the design of networks. In Section~\ref{sec:endogenous} we consider the nature of networks that would arise if players choose their networks and also choose efforts on battlefields. All the proofs not presented in the main text of the paper are collected in the Appendices. 

\section{Model}
\label{sec:model}

Two players, $1$ and $2$, compete on a set $B = \{1,\ldots,m\}$ of $m$ battlefields, each battlefield $k \in B$
of value $v^k > 0$ (which is for expositional simplicity common to both players). Each player $i \in \{1,2\}$ chooses a vector $\bm{e}_i = (e^k_i)_{k \in B} \in \mathbb{R}_{\geq 0}^B$ of non-negative efforts across the battlefields. Effort is costly and each player $i \in \{1,2\}$ faces a constant marginal cost of effort, $c_i > 0$, that can be different across the players.
For each player $i$ there is a network of non-negative effort spillovers between the battlefields, represented by the adjacency matrix $\bm{\rho}_i = (\rho_i^{k,l})_{k,l\in B}$, $\rho_i^{k,l} \geq 0$ for all $(k,l) \in B^2$.
Assignment of effort $e_i^k$ to battlefield $k$ by player $i$ results in spillover $\rho_i^{k,l} e_i^k$ to
battlefield $l \in B$. The vector of effort assignments by player $i$ to all the battlefields, $\bm{e}_i = (e_i^k)_{k \in B}$ results
in a vector of \emph{effective efforts} of $i$ due to network spillovers, $\bm{y}_i = (y_i^k)_{k \in B}$, with
$y_i^k = e_i^k + \sum_{l \in B\setminus\{k\}} \rho_i^{l,k} e_i^l$. In matrix notation
\begin{equation*}
\bm{y}_i = \left(\mathbf{I} + \bm{\rho}_i^T\right)\bm{e}_i,
\end{equation*}
where $\mathbf{I}$ denotes the identity matrix.

The probabilities of players winning the contest at battlefield $k \in B$, given the effective efforts assignments to $k$, $(y^k_1,y^k_2)$, are determined by a contest success function (CSF), $p : \mathbb{R}_{\geq 0}^2 \rightarrow [0,1]^2$.
The probability of player $i \in \{1,2\}$ winning the contest at battlefield $k \in B$ is $p_i(y^k_1,y^k_2)$.
Throughout most of the paper we focus on the Tullock contest success functions which have the form
\begin{equation*}
p_i(y_1,y_2) = \frac{(y_i)^{\gamma}}{(y_1)^{\gamma} + (y_2)^{\gamma}},
\end{equation*}
with $\gamma \in (0,1]$. The assumption that $\gamma\leq 1$ is fairly standard in the contest literature: it is made to ensure the existence of pure strategy Nash equilibrium. When $\gamma > 1$, payoffs cease to be concave even without spillovers and, depending on the costs and values of battlefields, equilibria in pure strategies may not exist; for a discussion of some of the issues that arise for large $\gamma$, see \cite{BayeKovenockdeVries1984} and~\cite{Ewerhart2015}.

Given the pair of efforts $(\bm{e}_1,\bm{e}_2)$, the expected payoff to player $i \in \{1,2\}$ is
\begin{equation}
\label{eq:payoff}
\payoff{i}(\bm{e}_1,\bm{e}_2)=\sum_{k\in B} v^k p_i^k(\bm{e}_1,\bm{e}_2) - c_i\sum_{k\in B} e_i^k,
\end{equation}
where
\begin{equation}
\label{eq:pcsf}
p_i^k(\bm{e}_1,\bm{e}_2) = p_i(y^k_1,y^k_2)
\end{equation}
is the probability of player $i$ winning battlefield $k$, given the efforts profile $(\bm{e}_1,\bm{e}_2)$.

A strategy profile $(\bm{e}^*_1,\bm{e}^*_2)\in \mathbb{R}^B_{\geq 0} \times \mathbb{R}^B_{\geq 0}$ is  a pure strategy Nash equilibrium of this interconnected conflict game  if, for any player $i\in \{1,2\}$ and any $(\bm{e}_1,\bm{e}_2)\in \mathbb{R}^B_{\geq 0} \times \mathbb{R}^B_{\geq 0}$,
\begin{equation*}
\payoff{i}(\bm{e}^*_1,\bm{e}^*_2) \geq \payoff{i}(\bm{e}_i,\bm{e}^*_{-i}).
\end{equation*}
We are interested in the properties of pure strategy Nash equilibria.

\section{Equilibrium in competing networks}
\label{sec:eqprop}

This section shows that there exists a unique equilibrium and provides a characterization of this equilibrium. The characterization yields a closed-form formula for equilibrium outcomes: the relative probability of player 1 versus player 2 winning a battlefield is equal to the product of the ratio of relative costs and the ratio of Bonacich centrality of this node in the two competing networks. 

We present a two-node example to bring out important aspects of strategic interaction in competing networks. 

\begin{example}
    Two-node Networks
\label{exampletwonodes}
\end{example}

There are two battlefields, $B = \{1,2\}$. Player $1$ has no spillovers across battlefields and player $2$ has a spillover $\lambda \geq 0$ from battlefield $1$ to battlefield $2$ (c.f.~Figure~\ref{fig:2nodes}). The adjacency matrices of spillovers for the two players are
\begin{equation*}
\bm{\rho}_{1} = \begin{bmatrix}
            0 & 0 \\
            0 & 0
            \end{bmatrix} \qquad\qquad\qquad
\bm{\rho}_{2} = \begin{bmatrix}
            0 & \lambda \\
            0 & 0
            \end{bmatrix}.
\end{equation*}
The CSF is Tullock CSF with $\gamma = 1$, the values of both battlefields $v^1 =v^2 = 1$. The costs are $c_1 > 0$ and $c_2 > 0$. In our numerical illustration below, 
we assume that $c_2>c_1$. 

\begin{figure}[htp]
\begin{center}
  \includegraphics{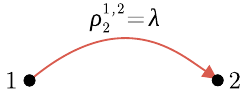}
\end{center}
  \caption{Spillovers for player $2$; no spillovers for player $1$.}
  \label{fig:2nodes}
\end{figure}

The payoffs to the players from strategy profile $(\bm{e}_1, \bm{e}_2)$ are
\begin{align*}
\payoff{1}(\bm{e}_1,\bm{e}_2) & = \frac{e_1^1}{e_1^1 + e_2^1} + \frac{e_1^2}{e_1^2 + \lambda e_2^1 + e_2^2} - (e_1^1 + e_1^2)c_1\\
\payoff{2}(\bm{e}_1,\bm{e}_2) & = \frac{e_2^1}{e_1^1 + e_2^1} + \frac{\lambda e_2^1 + e_2^2}{e_1^2 + \lambda e_2^1 + e_2^2} - (e_2^1 + e_2^2)c_2.
\end{align*}

There exists a unique equilibrium and it has the following structure: when $\lambda$ is small, both players choose positive efforts on both battlefields; when $\lambda$ has an intermediate value, player 1 continues to exert positive efforts on both battlefields, but player 2 allocates efforts only on battlefield 1; when $\lambda$ is large, player 1 only allocates efforts to battlefield 1 and abandons battlefield 2 to player 2, player 2 allocates positive efforts to node 1 only. The  computations are presented in the online appendix.  Figure~\ref{fig:effortsdiag} shows equilibrium efforts of the players as functions of the spillover $\lambda$.

\begin{figure}[htp]
\begin{center}
  \includegraphics[scale=0.9]{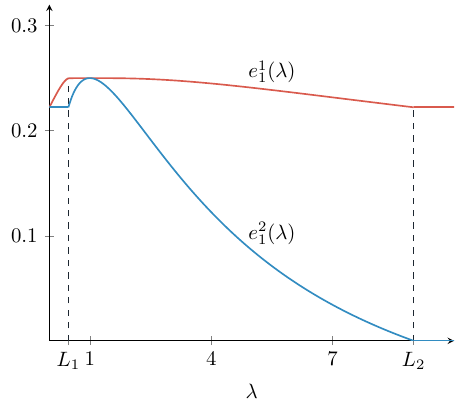}
  \includegraphics[scale=0.9]{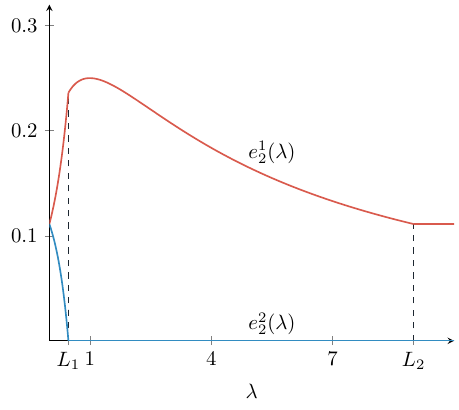}
\end{center}
  \caption{Equilibrium efforts as functions of the spillover $\lambda$ ($c_1 = 1$, $c_2 = 2$).}
  \label{fig:effortsdiag}
\end{figure}

  
$\hfill{\Box}$
This simple example shows that a player's efforts depend on their own network as well as the competing network;  large spillovers may induce a player to abandon a battlefield entirely. We now turn to a general analysis of interconnected contests.

We first establish the existence and uniqueness of equilibrium payoffs, probabilities of winning battlefields, and total efforts. 

\begin{theorem}
\label{th:exunique}
Let $p$ be the Tullock CSF with $\gamma \in (0,1]$. Pure strategy Nash equilibrium exists.
Moreover, equilibrium payoffs, probabilities of winning battlefields, and total efforts are unique.
\end{theorem}

The proof is presented in the appendix. 

Payoffs of players are discontinuous when one or more battlefields receive zero effective effort from both players; we employ standard methods involving bounded strategy spaces and take limits on these bounds to establish existence  (see e.g., \cite{XuZenouZhou2022}). On the issue of uniqueness, the main challenge is the possibility of corner equilibrium that arise naturally in a model with spillovers,   particularly those where only one player exerts  positive effective efforts at some battlefields (this is illustrated in Example~\ref{exampletwonodes}). 

A first observation is that equilibria in our model are \textit{interchangeable} and that in any equilibrium every battlefield receives positive effective effort from at least one player. Given this property, in any equilibrium the set of battlefields can be partitioned into three sets: those where both players exert positive effective efforts, set $B^{+}$, and those where only one of the player exerts positive effective effort, sets $A_1$ and $A_2$ (some of these parts may be empty). Next, we establish that this partition of battlefields is the same across all equilibria, and that equilibrium effective efforts of the players are unique on the set $B^{+}$. A direct consequence of these facts is uniqueness of probabilities of different players winning of winning the battlefields. Using this and interchangeability of equilibria, we obtain uniqueness of equilibrium payoffs and equilibrium total efforts.
 
It is worth noting that the result can be strengthened to uniqueness in efforts if $\gamma \in (0,1)$. In this case, the payoffs at each battlefield satisfies a form of Inada conditions: the marginal returns to effective efforts at a battlefield go to $\infty$ when the effective effort there goes to $0$. This, together with the fact that in equilibrium each battlefield receives positive effective effort from at least one player, immediately means that in equilibrium both players exert positive effective effort at every battlefield. This in turn implies that equilibrium effective efforts are unique. If, in addition, the matrices $\mathbf{I} + \bm{\rho}_i$ are non-singular, for $i \in \{1,2\}$, then equilibrium efforts are unique. Hence, generically, equilibria are unique for $\gamma \in (0,1)$. Matters are more complicated when $\gamma = 1$: now the Inada condition is not satisfied. It is then possible that some battlefields receive zero effective  efforts from one player and equilibrium effective efforts are not unique. We give an example of such a possibility in the (on-line) appendix.

Equipped with Theorem \ref{th:exunique}, we turn next to a characterization of this equilibrium. The standard approach to obtain equilibrium efforts in contest models involves solving the constrained optimization problem associated with the best-response conditions. If the solution is interior, the best response conditions result in a system of equations that characterizes equilibrium efforts. If the solution is corner, which is a typical situation in our model, the best response conditions feature both equations and inequalities. For expositional ease, and to develop intuitions, we will first state the result for the interior equilibrium; our arguments are general and also carry over to corner equilibrium; we state the general result below as Theorem~\ref{th:char} in this section. 

The marginal cost of effort for player $i$ is $c_i$: in an interior equilibrium, marginal return to effort on any node must be $c_i$. The marginal return is  equal to the sum of marginal returns to the contest on that node and the marginal returns to contests on linked nodes. Let us define $\mu_i^k$ as follows: 
\begin{equation*}
\mu_i^k = \left. \frac{\partial(p_i^k(\bm{e}_1,\bm{e}_2) v^k)}{\partial e_i^k} \right/ \frac{\partial(c_i e_i^k)}{\partial e_i^k}
\end{equation*}
In an interior equilibrium, the first order conditions imply that the following equation must hold in matrix form:
\begin{equation*}
\left(\mathbf{I} + \bm{\rho}_i\right) \bm{\mu}_i = \bm{1}
\end{equation*}
Rewriting, provided that $\mathbf{I} + \bm{\rho}_i$ is non-singular,
\begin{equation}
\label{eq:mu}
\bm{\mu}_i= \left(\mathbf{I} + \bm{\rho}_i\right)^{-1}\bm{1}.
\end{equation}
The marginal returns to a player from efforts at a node are proportional to the node's Katz-Bonacich network centrality in their own network, see e.g., \cite{bonacich2001eigenvector}, \cite{katz1953}.\footnote{In this case, Katz's $\alpha$ centrality has parameter $\alpha = -1$ and the external importance of every node is equal to $1$.}

Recall, that given two vectors $\bm{z}\in \mathbb{R}^S$ and $\bm{y} \in \mathbb{R}^S$, $\bm{z} \odot \bm{y} = (z_i y_i)_{i\in S}$ is the Hadamard product (i.e. entry-wise product) of $\bm{z}$ and $\bm{y}$ and $\bm{z} \oslash \bm{y} = (z_i/y_i)_{i\in S}$ is the Hadamard division (i.e. entry-wise division) of $\bm{z}$ by $\bm{y}$ when $y_i \neq 0$.

\begin{theorem}
\label{th:char1}
Let $p$ be the Tullock CSF with $\gamma \in (0,1]$ and assume $\mathbf{I} + \bm{\rho}_i$ is non-singular, for $i \in \{1,2\}$. In an interior equilibrium $(\bm{e}_1,\bm{e}_2)$
\begin{equation*}
\begin{aligned}
\bm{e}_i & = \frac{\gamma}{c_i} \left(\mathbf{I} + {\bm{\rho}_i}^T\right)^{-1} \left(\bm{p}_1 \odot\bm{p}_2 \odot \bm{v} \oslash \bm{\mu}_i\right)
\end{aligned}
\end{equation*}
where
\begin{equation*}
\begin{aligned}
& p^k_i = \frac{(\mu^k_{-i} c_{-i})^{\gamma}}{(\mu^k_1 c_1)^{\gamma} + (\mu^k_2 c_2)^{\gamma}} \textrm{, for all $k \in B$}
\end{aligned}
\end{equation*}
\end{theorem}

A key insight of Theorem~\ref{th:char1} is that the ratio of probabilities of winning a node $k \in B$ equals
\begin{equation}
\frac{p^k_i}{p^k_j} = \left(\frac{\mu^k_{j} c_{j}}{\mu^k_i c_i}\right)^{\gamma},
\end{equation}
at any interior equilibrium, i.e., the ratio of the Katz-Bonacich centrality of this node in the two networks times the ratio of costs (adjusted for the coefficient of the Tullock contest function). This ratio in battle $k$ exceeds one if and only if  $\mu^k_i c_i < \mu^k_j c_j$, meaning that player $i$'s marginal cost,  adjusted by  network centrality $\mu^k_i$ in battle $k$, is lower than that of $j$.   Since the sum of the winning probabilities equals one, this ratio fully determines both the winning probabilities of each player and their equilibrium efforts in  Theorem \ref{th:char1}. This expression on relative probabilities also helps us obtain analytical solutions for the equilibrium  based on the model's primitives (assuming the existence of an interior equilibrium).

Let us elaborate on the ideas underlying this result. When matrices $\mathbf{I} + \bm{\rho}_i$ are non-singular for $i \in \{1,2\}$, our game can be transformed to a game in the space of effective efforts, where the players choose effective efforts and the costs are redefined to map the cost of effective efforts to the cost of the associated real efforts. In this new game, the set of strategies of player $i$, $Y_i$, is the result of the linear mapping $\mathbf{I}+\bm{\rho}_i^T$ of the first quadrant $\mathbb{R}^B_{\geq 0}$ (the strategy set of player $i$ in the original game),
\begin{equation*}
Y_i = \{\bm{y_i} = \left(\mathbf{I}+\bm{\rho}_i^T\right) \bm{e_i}, \textrm{ where }  \bm{e_i} \in \mathbb{R}^B_{\geq 0} \}.
\end{equation*}
The payoff to player $i$ from effective effort profile $(\bm{y}_1,\bm{y}_2) \in Y_1 \times Y_2$ in the new game is
\begin{equation*}
\begin{aligned}
\widetilde{\mathit{\Pi}}_i(\bm{y}_1,\bm{y}_2) = \sum_{k\in B} p_i(y^k_1,y^k_2) v^k - c_i \bm{1}^T \left(\mathbf{I}+\bm{\rho}_i^T\right)^{-1}\bm{y}_i.
\end{aligned}
\end{equation*}
Using~\eqref{eq:mu}, the payoff to player $i$ can be rewritten as
\begin{equation*}
\widetilde{\mathit{\Pi}}_i(\bm{y}_1,\bm{y}_2) = \sum_{k\in B} p_i(y^k_1,y^k_2) v^k - c_i \bm{\mu}_i^T \bm{y}_i = \sum_{k\in B} p_i(y^k_1,y^k_2) v^k - \sum_{k \in B} c_i \mu_i^k y^k_i.
\end{equation*}
This shows that the modified game is a game on multiple battlefields where the costs of each player $i$ are heterogeneous across the battlefields and the cost at a battlefield $k$ of $i$ is equal to $c_i \mu^k_i$.

In addition, the network restricts the space of effective efforts that player $i$ can choose in this game. If the restriction is not binding, we have an interior equilibrium where the probability of winning battlefield $k$ by player $i$ is equal to
\begin{equation*}
\frac{(c_{-i}\mu^k_{-i})^{\gamma}}{(c_{1}\mu^k_{1})^{\gamma} + (c_{2}\mu^k_{2})^{\gamma}}. 
\end{equation*} 
In this case, the quantities $\bm{\mu}_i$ of player $i$ coincide with Katz-Bonaccich centralities. The equilibrium effective efforts in our game are the same as the equilibrium of a game on multiple battlefields with heterogeneous costs of players across the battlefields and with unrestricted strategy spaces (where each player can choose any non-negative effort at each battlefield).

We next present the general characterization result that covers corner equilibrium. This requires us to develop some new notation. For each player $i \in \{1,2\}$ there exists a set of battlefields, $P_i \subseteq B$,
at which $i$ exerts positive effort.  Given a vector $\bm{z} \in \mathbb{R}^B$ and a set $S \subseteq B$, we will use $\bm{z}^S = (z^k)_{k \in S}$ to denote the vector $\bm{z}$ restricted to the entries in $S$. Similarly, given a matrix $\bm{a} \in \mathbb{R}^{B\times B}$ and two sets $S\subseteq B$ and $T\subseteq B$, we will use $\bm{a}^{S,T} = (a^{k,l})_{k\in S,l\in T}$ to denote matrix $\bm{a}$ restricted to the entries in $S\times T$. In particular, matrix $\bm{\rho}_i^{S,T}$ is the adjacency matrix of the spillovers from the battlefields in $S$ to the battlefields in $T$ of player $i \in \{1,2\}$. Given a set $S \subseteq B$, we will also use $-S = B\setminus S$ to denote the complement of $S$.

\begin{theorem}
\label{th:char}
Let $p$ be the Tullock CSF with $\gamma \in (0,1]$ and assume $\mathbf{I} + \bm{\rho}_i$ and all its submatrices are non-singular, for $i \in \{1,2\}$. In equilibrium, every battlefield receives positive effective effort from at least one player $i \in \{1,2\}$.
A strategy profile $(\bm{e}_1,\bm{e}_2)$ is a Nash equilibrium if and only if there exist sets of battlefields $P_1\subseteq B$
and $P_2\subseteq B$ such that, for all $i \in \{1,2\}$, 

\begin{equation}
\label{eq:char:e}
\begin{aligned}
\bm{e}_i^{-P_i} & = \bm{0},\\
\bm{e}_i^{P_i} & = \frac{\gamma}{c_i} \left(\mathbf{I} + {\bm{\rho}_i^{P_i,P_i}}^T\right)^{-1} \left(\bm{p}_1^{P_i} \odot\bm{p}_2^{P_i} \odot \bm{v}^{P_i} \oslash \bm{\mu}_i^{P_i}\right),
\end{aligned}
\end{equation}
and $\bm{\mu}_i \geq 0$ satisfies
\begin{equation}
\label{eq:char:mu}
\begin{aligned}
& \left(\mathbf{I} + \bm{\rho}_i^{P_i,P_i}\right) \bm{\mu}_i^{P_i} + \bm{\rho}_i^{P_i,-P_i} \bm{\mu}_i^{-P_i} = \bm{1}\\
& {\bm{\rho}_i^{P_i,-P_i}}^T \left(\mathbf{I} + {\bm{\rho}_i^{P_i,P_i}}^T\right)^{-1} \left(\bm{p}_1^{P_i} \odot \bm{p}_2^{P_i} \odot \bm{v}^{P_i} \oslash \bm{\mu}_i^{P_i}\right) & = \bm{p}_1^{-P_i} \odot \bm{p}_2^{-P_i} \odot \bm{v}^{-P_i} \oslash \bm{\mu}_i^{-P_i}
\end{aligned}
\end{equation}
and 
\begin{equation}
\label{eq:char:pr}
p^k_i = \frac{(\mu^k_{-i} c_{-i})^{\gamma}}{(\mu^k_1 c_1)^{\gamma} + (\mu^k_2 c_2)^{\gamma}} \textrm{, for all $k \in B$}.
\end{equation}

Moreover, the equilibrium total effort of player $i$ is
\begin{equation}
\label{eq:char:total}
\sum_{k \in B} e^k_i = \frac{\gamma}{c_i} \left(\sum_{k\in B} p_1^k p_2^k v^k\right)
\end{equation}
and the payoff to player $i$ is
\begin{equation}
\label{eq:char:payoff}
\payoff{i}(\bm{e}_1,\bm{e}_2) = \sum_{k \in B} \left(p^k_i\left(1 - \gamma p^k_{-i}\right) v^k\right).
\end{equation}
\end{theorem}

\begin{proof} A strategy profile $(\bm{e}_1,\bm{e}_2)$ with the associated effective efforts profile, $(\bm{y}_1,\bm{y}_2)$, is a Nash equilibrium of the game if and only if, for any for any $i\in \{1,2\}$ and any $k \in B$ it satisfies
\begin{equation}
\label{eq:kkt:h1}
\begin{aligned}
\frac{\partial \payoff{i}}{\partial e_i^k} = \frac{\gamma p_i^k(1-p_i^k)}{y_i^k}v^k + \sum_{l \in B\setminus \{k\}} \rho^{kl}_i
\frac{\gamma p_i^l(1-p_i^l)}{y_i^l} v^l-c_i & = 0 \textrm{, if $e_i^k > 0$,} \\
\frac{\partial \payoff{i}}{\partial e_i^k}=\frac{\gamma p_i^k(1-p_i^k)}{y_i^k}v^k + \sum_{l \in B\setminus \{k\}} \rho^{kl}_i
\frac{\gamma p_i^l(1-p_i^l)}{y_i^l} v^l-c_i & \leq 0 \textrm{, if $e_i^k = 0$,}
\end{aligned}
\end{equation}
where the winning probability $p_i^k$ is given in~\eqref{eq:pcsf}.
The ``only if'' direction is clear as~\eqref{eq:kkt:h1} is the first-order condition of player $i$ with respect to $e_i^k$ (the payoff of $i$ is continuously differentiable in $\bm{e}_i$ as $\bm{e}$ is of type $S^1$).
The ``if'' direction follows because the payoff function of player $i$ is concave and the set of admissible efforts is convex.
Therefore the local optimality condition given in~\eqref{eq:kkt:h1} implies global optimality.
By Lemma~\ref{lemma:aux:2}, in the appendix, $y^k_1 + y^k_2 > 0$, for all $k \in B$. By Lemma~\ref{lemma:uniqeff} in the appendix, in the case of $\gamma \in (0,1)$, $y^k_i > 0$ for all $i \in \{1,2\}$ and all $k \in B$. In the case of $\gamma = 1$ the derivatives are finite for all $y^k_i \geq 0$ because $p_i^k/y_i^k = 1/(y^k_1 + y^k_2)$. Thus all the derivatives are well defined around every equilibrium and the payoff functions are differentiable in the neighbourhood of any equilibrium $(\bm{y}_1,\bm{y}_2)$.
        
Let
\begin{equation*}
\mu_i^k = \left. \frac{\partial(p_i^k v^k)}{\partial e_i^k} \right/ \frac{\partial(c_i e_i^k)}{\partial e_i^k}
        = \frac{\partial(p_i^k v^k)}{\partial(c_i e_i^k)}
\end{equation*}
be the marginal rate of substitution between the expected reward from winning the prize and the cost of effort.
At strategy profile $(\bm{e}_1,\bm{e}_2)$,
\begin{equation*}
\mu_i^k = \frac{\gamma (y_{i}^k)^{\gamma-1} (y_{-i}^k)^{\gamma}}{(({y}_1^k)^{\gamma} + ({y}_2^k)^{\gamma})^2 c_i}v^k,
\end{equation*}
which in the case of $y_i^k > 0$ can be written as
\begin{equation}
\label{eq:muhdef}
\mu_i^k = \frac{\gamma (y_{i}^k)^{\gamma-1} (y_{-i}^k)^{\gamma}}{(({y}_1^k)^{\gamma} + ({y}_2^k)^{\gamma})^2 c_i}v^k = \frac{\gamma p_i^k(1-p_i^k)}{{y}_i^k c_i}v^k = \frac{\gamma p_1^k p_2^k}{{y}_i^k c_i}v^k,
\end{equation}
where $p_i^k = p_i({y}_1^k,{y}_2^k)$. 
Using that, \eqref{eq:kkt:h1} can be rewritten as
\begin{equation}
\label{eq:muh:0}
\begin{aligned}
\mu_i^k + \sum_{l \in B\setminus \{k\}} \rho^{kl}_i \mu_i^l = 1\ \textrm{, if $e_i^k > 0$,} \\
\mu_i^k + \sum_{l \in B\setminus \{k\}} \rho^{kl}_i \mu_i^l \leq 1\ \textrm{, if $e_i^k = 0$.}
\end{aligned}
\end{equation}

The probability of winning battle $k$ by player $i$ at effective efforts profile $(\bm{y}_1,\bm{y}_2)$ is
\begin{equation*}
p^k_i = \frac{(y^k_i)^{\gamma}}{(y^k_1)^{\gamma} + (y^k_2)^{\gamma}} = \frac{1}{1 + \left(\frac{y^k_{-i}}{y^k_{i}}\right)^{\gamma}}
\end{equation*}
which, by~\eqref{eq:muhdef}, in the case of $y_i^k > 0$ is equal to
\begin{equation}
\label{eq:probabilityh}
p^k_i = \frac{1}{1 + \left(\frac{\mu^k_{i} c_i}{\mu^k_{-i} c_{-i}}\right)^{\gamma}} = \frac{(\mu^k_{-i} c_{-i})^{\gamma}}{(\mu^k_1 c_1)^{\gamma} + (\mu^k_2 c_2)^{\gamma}}.
\end{equation}
In the case of $y_i^k = 0$, $y_{-i}^k > 0$, and $\gamma = 1$, 
\begin{equation*}
\mu_i^k = \frac{y_{-i}^k}{({y}_1^k + {y}_2^k)^2 c_i}v^k = \frac{v^k}{y_{-i}^k c_i},
\end{equation*}
\begin{equation*}
\mu_{-i}^k = \frac{y_{i}^k}{({y}_1^k + {y}_2^k)^2 c_i}v^k = 0,
\end{equation*}
$p_i^k = 0$, and $p_{-i}^k = 1$. Hence~\eqref{eq:probabilityh} is valid in that case as well.

Inserting this in~\eqref{eq:muhdef} and solving for $y^k_i$ we obtain
\begin{equation*}
y^k_i = \frac{\gamma(\mu^k_{i}c_i)^{\gamma-1}(\mu^k_{-i} c_{-i})^{\gamma}}{((\mu^k_1 c_1)^{\gamma} + (\mu^k_2 c_2)^{\gamma})^2}v^k.
\end{equation*}
Abusing the notation slightly in the case of $\mu_i^k = 0$, this can be written more concisely as
\begin{equation*}
y^k_i = \frac{\gamma p_1^k p_2^k}{\mu_i^k c_i}v^k.
\end{equation*}
In matrix form,
\begin{equation*}
\bm{y}_i = \frac{\gamma}{c_i}
\begin{bmatrix}
\frac{p_1^1 p_2^1 v^1 }{\mu^1_1} \\
\vdots\\
\frac{p_1^m p_2^m v^m}{\mu^m_1}
\end{bmatrix}
= \frac{\gamma}{c_i}
\begin{bmatrix}
p_1^1 p_2^1 v^1 \\
\vdots\\
p_1^m p_2^m v^m
\end{bmatrix} \oslash \bm{\mu}_i
= \frac{\gamma}{c_i} \bm{p}_1 \odot \bm{p}_2 \odot \bm{v} \oslash \bm{\mu}_i
\end{equation*}
and, given the definition of $\bm{y}_i$ in terms of $\bm{e}_i$,
\begin{equation}
\label{eq:effort}
\bm{e}_i = \left(\mathbf{I}+\bm{\rho}_i^T\right)^{-1} \bm{y}_i = \frac{\gamma}{c_i} \left(\mathbf{I}+\bm{\rho}_i^T\right)^{-1} \left(\bm{p}_1 \odot \bm{p}_2 \odot \bm{v} \oslash \bm{\mu}_i\right).
\end{equation}
Introducing slack variables $s_i^k$ (for $k \in B$), the system of inequalities~\eqref{eq:muh:0} can be written as
\begin{equation*}
\mu_i^k + \sum_{l \in B\setminus \{k\}} \rho^{kl}_i \mu_i^l = 1 - s^k_i
\end{equation*}
where, for all $k \in B$, $s^k_i \geq 0$ and $s^k_i = 0$, if $e^k_i > 0$.
In matrix form,
\begin{equation}
\label{eq:muh}
\left(\mathbf{I}+\bm{\rho}_i\right)\bm{\mu}_i = \mathbf{1} - \bm{s}_i,
\implies
\bm{\mu}_i = \left(\mathbf{I}+\bm{\rho}_i\right)^{-1}(\mathbf{1} - \bm{s}_i).
\end{equation}

Let $P_i = \{k \in B : e_i^k > 0\}$ be the set of battlefields receiving positive real effort from player $i$ under $\bm{e}_i$ and let $-P_i = B\setminus \{P_i\}$ be the set of battlefields receiving zero real effort from player $i$ under $\bm{e}_i$.
By~\eqref{eq:muh}
\begin{equation*}
\left(\mathbf{I}^{P_i,P_i} + \bm{\rho}_i^{P_i,P_i}\right)\bm{\mu}_i^{P_i} + \bm{\rho}_i^{P_i,-P_i} \bm{\mu}_i^{-P_i}  = \mathbf{1}.
\end{equation*}
Moreover, since $\bm{e}_i^{-P_i} = \bm{0}$ and $\bm{y}_i = (\mathbf{I} + \bm{\rho}_i^T)\bm{e}_i$ so
\begin{align*}
\bm{y}_i^{P_i} & = \left(\mathbf{I}^{P_i,P_i}+{\bm{\rho}_i^{P_i,P_i}}^T\right) \bm{e}_i^{P_i}\\
\bm{y}_i^{-P_i} & = {\bm{\rho}_i^{P_i,-P_i}}^T \bm{e}_i^{P_i}
\end{align*}
Hence, using~\eqref{eq:effort},
\begin{equation*}
\bm{e}_i^{P_i} = \frac{\gamma}{c_i} \left(\mathbf{I}^{P_i,P_i}+{\bm{\rho}_i^{P_i,P_i}}^T\right)^{-1} \left(\bm{p}_1^{P_i} \odot \bm{p}_2^{P_i} \odot \bm{v}^{P_i} \oslash \bm{\mu}_i^{P_i}\right)
\end{equation*}
and
\begin{align*}
\bm{p}_1^{-P_i} \odot \bm{p}_2^{-P_i} \odot \bm{v}^{-P_i} \oslash \bm{\mu}_i^{-P_i} & = \frac{c_i}{\gamma}{\bm{\rho}_i^{P_i,-P_i}}^T \bm{e}_i^{P_i}\\
                                        & = {\bm{\rho}_i^{P_i,-P_i}}^T\left(\mathbf{I}^{P_i,P_i}+{\bm{\rho}_i^{P_i,P_i}}^T\right)^{-1} \left(\bm{p}_1^{P_i} \odot \bm{p}_2^{P_i} \odot \bm{v}^{P_i} \oslash \bm{\mu}_i^{P_i}\right).
\end{align*}

Using~\eqref{eq:muh}, the total effort of player $i \in \{1,2\}$ in Nash equilibrium $(\bm{e}_1,\bm{e}_2)$ is
\begin{equation*}
\sum_{k \in B} e^k_i = \mathbf{1}^T \bm{e}_i = \left(\mathbf{1} - \bm{s}_i\right)^T\bm{e}_i,
\end{equation*}
as $s_i^k = 0$, if $e_i^k > 0$, and $e_i^k = 0$, if $s_i^k > 0$. Hence
\begin{align*}
\sum_{k \in B} e^k_i & = \frac{\gamma}{c_i}
\underbrace{ (\mathbf{1} - \bm{s}_i)^T
\left(\mathbf{I}+\bm{\rho}_i^T\right)^{-1}}_{= \bm{\mu}_i^T} \bm{p}_1 \odot \bm{p}_2 \odot \bm{v} \oslash \bm{\mu}_i
= \frac{\gamma}{c_i} \left(\sum_{k\in B} p^k_1 p^k_2 v^k\right)\\
& = \frac{\gamma}{c_i} \sum_{k\in B} \frac{(\mu^k_{1} \mu^k_{2} c_{1} c_{2})^{\gamma}}{((\mu^k_1 c_1)^{\gamma} + (\mu^k_2 c_2)^{\gamma})^2} v^k.
\end{align*}

The equilibrium payoff to player $i$ is
\begin{align*}
\payoff{i}(\bm{e}_1,\bm{e}_2) & = \sum_{k \in B} p^k_i v^k - c_i \bm{1}^T \bm{e}_i
= \sum_{k \in B} \left(p^k_i\left(1 - \gamma p^k_{-i}\right) v^k\right)\\
& = \sum_{k \in B} \frac{(c_{-i}\mu^k_{-i})^{\gamma}((c_1 \mu^k_1)^{\gamma} + (c_2 \mu^k_2)^{\gamma} - \gamma (c_i \mu^k_i)^{\gamma})) v^k }{((c_1 \mu^k_1)^{\gamma} + (c_2 \mu^k_2)^{\gamma})^2}
\end{align*}

This completes the proof.
\end{proof}

If an equilibrium strategy of player $i$ is not interior, the restriction of the set of strategies of $i$ in the game played in the space of effective efforts is binding. In this case, the quantities $\bm{\mu}_i$ do not coincide with Bonacich-Katz centrality. An equilibrium strategy of player $i$ is a best response of $i$ within the restricted strategy space $Y_i$ to an equilibrium strategy of the other player, which is a strategy in the restricted space $Y_{-i}$. Because of that, in the case of corner equilibrium, quantities $\mu_i$ depend not only on $\bm{\rho}_i$ but also on $\bm{\rho}_{-i}$.

Thus the network affects the contest between the two players in two ways: by scaling the costs of the players with Bonacich-Katz centralities and by restricting the strategy spaces of the players in the space of effective efforts. In the case of interior equilibrium this second effect is immaterial but in the case of corner equilibrium both effects matter.

We now present three examples that illustrate the theorem. The first example highlights the role of networks in amplifying advantages. 

\begin{example}
Star vs. Cycle Network
\end{example}

Consider a scenario with $m \geq 4$ battlefields, $B = \{1,2,3,\ldots,m\}$.
Suppose that the network of spillovers of player $1$, $\bm{\rho}_1$, is an undirected star with centre $1$ and the network of spillovers of player $2$, $\bm{\rho}_2$, is an undirected cycle (c.f. Figure~\ref{fig:starandcycle}).
Assume also that $m$ is not divisible by 3.\footnote{When $m$ is divisible by $3$, there is multiplicity of equilibria, because $(\mathrm{I} + \bm{\rho}_2)$ is non-invertible. All the equilibria result in the same effective efforts profile, the same total efforts of the two players and the same payoffs for the two players.}
In both cases, for simplicity, the magnitude of spillovers at each link is equal to $1$.
The CSF is Tullock CSF with $\gamma = 1$, the values of all battlefields are equal to $1$. The costs are equal to $1$, $c_1 = c_2 = 1$.

\begin{figure}[htp]
\begin{center}
  \includegraphics[scale=0.7]{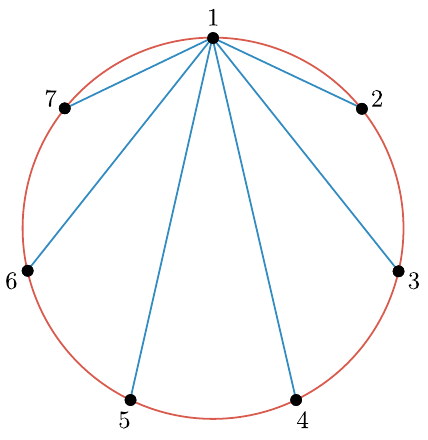}
\end{center}
  \caption{An undirected star with centre $1$ and an undirected cycle over $m = 7$ nodes.}
  \label{fig:starandcycle}
\end{figure}

Let $P_1 = \{1\}$ and $P_2 = B$. Solving the system of equations defined by~\eqref{eq:char:mu} and~\eqref{eq:char:pr} in~Theorem~\ref{th:char} we compute the quantities $\bm{\mu}_1$ and $\bm{\mu}_2$:
\begin{equation*}
\mu_1^1 = \ldots = \mu_1^m = \frac{1}{m},\qquad
\mu_2^1 = \ldots = \mu_2^m = \frac{1}{3}.
\end{equation*}
Using~\eqref{eq:char:pr} we obtain the winning probabilities:
\begin{equation*}
p_1^1 = \ldots = p_1^m = \frac{m}{m+3},\qquad
p_2^1 = \ldots = p_2^m = \frac{3}{m+3}.
\end{equation*}
Notice that the ratio of the probabilities for battlefield $k\in B$, $p_1^k/p_2^k$, is equal to the inverse ratio of the corresponding quantities $\mu_2^k/\mu_1^k = m/3$.

Inserting the $\bm{\mu}$'s and the probabilities into~\eqref{eq:char:e} we obtain a system of equations for finding the efforts. Solving it we obtain
\begin{equation*}
e_1^1 = \frac{3m^2}{(m+3)^2},\qquad
e_1^2 = \ldots = e_1^m = 0,\qquad
e_2^1 = \ldots = e_2^m = \frac{3m}{(m+3)^2}.
\end{equation*}
This is a valid strategy profile so we have a Nash equilibrium.

The equilibrium total efforts of the two players are equal to
\begin{equation*}
\sum_{k \in B} e^k_i = \frac{3m^2}{(m+3)^2}.
\end{equation*}
The equilibrium payoffs are
\begin{equation*}
\payoff{1}(\bm{e}_1,\bm{e}_2) = \frac{m^2}{(m+3)^2}m, \qquad\qquad \payoff{2}(\bm{e}_1,\bm{e}_2) = \frac{9}{(m+3)^2}m
\end{equation*}
and the ratio of the two equilibrium payoffs is
\begin{equation*}
\frac{\payoff{1}(\bm{e}_1,\bm{e}_2)}{\payoff{2}(\bm{e}_1,\bm{e}_2)} = \frac{m^2}{9}.
\end{equation*}
This example illustrates that power of the network topology: it can generate  unbounded advantage for the player with the star network.\footnote{Interestingly, this advantage arises even though the total magnitude of spillovers in the cycle network is higher than in the star network.}
\hfill $\Box$

\strut

In the cycle network, all nodes are symmetric. The second example involves a star network competing with a line network -- it helps bring out the role of network asymmetries in shaping contests in competing networks. 

\begin{example}
Star versus Line Network 
\end{example}

Consider an example with $m=5$ where $\bm{\rho}_1$  is an undirected cycle, like in the previous example, and $\bm{\rho}_2$ is an undirected line $1-2-3-4-5$. (c.f. Figure~\ref{fig:starandline})).
In both cases the magnitude of spillovers at each link is equal to $1$.
The CSF is Tullock CSF with $\gamma = 1$, the values of all battlefields are equal to $1$. The costs are also equal to $1$, $c_1 = c_2 = 1$.

\begin{figure}[htp]
\begin{center}
 \includegraphics[scale=0.6]{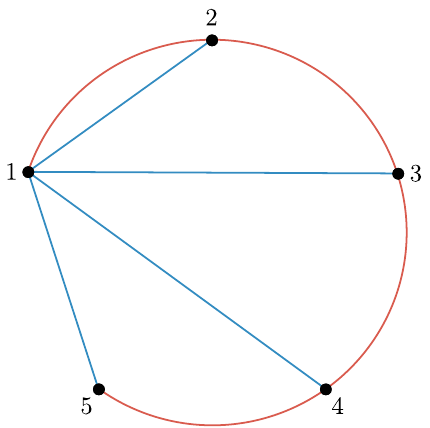}
\end{center}
 \caption{An undirected star with centre $1$ and an undirected line over $m = 5$ nodes.}
 \label{fig:starandline}
\end{figure}

Let $P_1 = \{1\}$ and $P_2 = \{2,4\}$. Solving the system of equations defined by~\eqref{eq:char:mu} and~\eqref{eq:char:pr} in~Theorem~\ref{th:char} we compute the quantities $\bm{\mu}_1$ and $\bm{\mu}_2$:
\begin{align*}
\mu_1^1 & = \mu_1^2 = \mu_1^4 = \mu_1^5 = \frac{8}{41},\qquad \mu_1^3 = \frac{9}{41},\\
\mu_2^1 & = \mu_2^2 = \mu_2^4 = \mu_2^5 = \frac{16}{41},\qquad \mu_2^3 = \frac{9}{41}.
\end{align*}
Using~\eqref{eq:char:pr} we obtain the winning probabilities:
\begin{align*}
p_1^1 & = p_1^2 = p_1^4 = p_1^5 = \frac{2}{3},\qquad p_1^3 = \frac{1}{2},\\
p_2^1 & = p_2^2 = p_2^4 = p_2^5 = \frac{1}{3},\qquad p_1^3 = \frac{1}{2}.
\end{align*}
Notice that the ratio of the probabilities for all battlefields $k\in B$, $p_1^k/p_2^k$, is equal to the inverse ratio of the corresponding quantities, $\mu_2^k/\mu_1^k = 2$ (in the case of $k\neq 3$) and $\mu_2^k/\mu_1^k = 1$ (in the case of $k = 3$).

Inserting the $\bm{\mu}$'s and the probabilities into~\eqref{eq:char:e} we obtain a system of equations for finding the efforts. Solving it we obtain
\begin{equation*}
e_1^1 = \frac{41}{36},\qquad
e_1^2 = \ldots = e_1^5 = 0,\qquad
e_2^2 = e_2^4 = \frac{41}{72}, \qquad e_2^1 = e_2^3 = e_2^5 = 0.
\end{equation*}
This is a valid strategy profile so we have a Nash equilibrium.

The equilibrium total efforts of the two players are both equal to $41/36$.
The equilibrium payoffs are
\begin{equation*}
\payoff{1}(\bm{e}_1,\bm{e}_2) = \frac{73}{36}, \qquad 
\payoff{2}(\bm{e}_1,\bm{e}_2) = \frac{25}{36},
\end{equation*}
and the ratio of the two equilibrium payoffs is $73/25$. Although the two networks have the same total magnitude of spillovers,  the star topology is  nonetheless more beneficial than the line topology.
\hfill$\Box$

\strut

The next example illustrates the role of competing networks in shaping equilibrium outcomes. 

\begin{example}
Competing hub-spoke networks
\end{example}
    
There are $m=10$ battlefields. Suppose that player $1$ has spillovers $\lambda \geq 0$ from battlefield $1$ to all remaining battlefields and player $2$ has spillovers $\lambda \geq 0$ from battlefield $2$ to all remaining battlefields. We assume that the Tullock contest success function has $\gamma = 1$. The value of every battlefield is equal to $1$. The costs are equal to $1$, $c_1 = c_2 = 1$.  The network is illustrated in Figure \ref{figuretwohubs}. 

There are three types of equilibria, depending on the magnitude of spillovers $\lambda$.

Let $\lambda \in [0,L_1)$, where $L_1 = 1/n$. Let $P_1 = P_2 = B$. Solving the system of equations defined by~\eqref{eq:char:mu} and~\eqref{eq:char:pr} in~Theorem~\ref{th:char} we compute the quantities $\bm{\mu}_1$ and $\bm{\mu}_2$:
\begin{equation*}
\mu_1^1 = \mu_2^2 = 1 - (n - 1)\lambda,\qquad
\mu_1^2 = \mu_2^1 = 1,\qquad
\mu_1^k = \mu_2^k = 1.
\end{equation*}
Using~\eqref{eq:char:pr} we obtain the winning probabilities:
\begin{equation*}
p_1^1 = p_2^2 = \frac{1}{2 - (n - 1)\lambda},\qquad
p_1^2 = p_2^1 = \frac{1 - (n - 1)\lambda}{2 - (n - 1)\lambda},\qquad
p_1^k = p_2^k = \frac{1}{2}.
\end{equation*}
Inserting the $\bm{\mu}$'s and the probabilities into~\eqref{eq:char:e} we obtain a system of equations for finding the efforts. Solving it we obtain
\begin{align*}
e_1^1 & = e_2^2 = \frac{1}{(2 - (n - 1)\lambda)^2},\qquad
e_1^2 = e_2^1 = \frac{1 - n\lambda}{(2 - (n - 1)\lambda)^2},\\
e_1^k & = e_2^k = \frac{1}{4} - \frac{\lambda}{(2 - (n - 1)\lambda)^2}.
\end{align*}

Let $\lambda \in [L_1,L_2)$, where $\lambda = L_2$ solves $\lambda^3 (n - 2) + \lambda^2(2n + 3) + \lambda(n - 4) = 1$,
and is positive, that is
\begin{equation*}
L_2 = \frac{2\sqrt{n^2 + 30n - 15}\sin\left(\frac{\arcsin\left(\frac{n^3 - 63n^2 - 117n + 135}{(n^2 + 30n - 15)^{\frac{3}{2}}}\right)}{3} + \frac{\pi}{3}\right)}{3(n - 2)} - \frac{2n + 3}{3(n - 2)}
\end{equation*}
(e.g. in the case of $n = 10$, $L_2 \approx 0.114453$).
Let $P_1 = B\setminus\{2\}$ and $P_2 = B\setminus \{1\}$. Solving the system of equations defined by~\eqref{eq:char:mu} and~\eqref{eq:char:pr} in~Theorem~\ref{th:char} we compute the quantities $\bm{\mu}_1$ and $\bm{\mu}_2$:
\begin{equation*}
\mu_1^1 = \mu_2^2 = \frac{1 - (n - 2)\lambda}{2},\qquad
\mu_1^2 = \mu_2^1 = \frac{1 - (n - 2)\lambda}{2\lambda},\qquad
\mu_1^k = \mu_2^k = 1.
\end{equation*}
Using~\eqref{eq:char:pr} we obtain the winning probabilities:
\begin{equation*}
p_1^1 = p_2^2 = \frac{1}{1+\lambda},\qquad
p_1^2 = p_2^1 = \frac{\lambda}{1+\lambda},\qquad
p_1^k = p_2^k = \frac{1}{2}.
\end{equation*}
Inserting the $\bm{\mu}$'s and the probabilities into~\eqref{eq:char:e} we obtain a system of equations for finding the efforts. Solving it we obtain
\begin{align*}
e_1^1 & = e_2^2 = \frac{2\lambda}{(1+\lambda)^2 (1 - (n - 2)\lambda)^2},\qquad
e_1^2 = e_2^1 = 0,\\
e_1^k & = e_2^k = \frac{1}{4} - \frac{2\lambda^2}{(1+\lambda)^2 (1 - (n - 2)\lambda)^2}.
\end{align*}

Let $\lambda \geq L_2$. Let $P_1 = \{1\}$ and $P_2 = \{2\}$. Solving the system of equations defined by~\eqref{eq:char:mu} and~\eqref{eq:char:pr} in~Theorem~\ref{th:char} we compute the quantities $\bm{\mu}_1$ and $\bm{\mu}_2$:
\begin{align*}
\mu_1^1 & = \mu_2^2 = \frac{4\lambda}{2(n+2)\lambda + (n-2)(1+\lambda^2)},\qquad
\mu_1^2 = \mu_2^1 = \frac{4}{2(n+2)\lambda + (n-2)(1+\lambda^2)},\\
\mu_1^k & = \mu_2^k = \frac{(1+\lambda)^2}{\lambda(2(n+2)\lambda + (n-2)(1+\lambda^2))}.
\end{align*}
Using~\eqref{eq:char:pr} we obtain the winning probabilities:
\begin{equation*}
p_1^1 = p_2^2 = \frac{1}{1+\lambda},\qquad
p_1^2 = p_2^1 = \frac{\lambda}{1+\lambda},\qquad
p_1^k = p_2^k = \frac{1}{2}.
\end{equation*}
Inserting the $\bm{\mu}$'s and the probabilities into~\eqref{eq:char:e} we obtain a system of equations for finding the efforts. Solving it we obtain
\begin{equation*}
e_1^1 = e_2^2 = \frac{n-2}{4} + \frac{2\lambda}{(1+\lambda)^2},\qquad
e_1^2 = e_2^1 = 0,\qquad
e_1^k = e_2^k = 0,
\end{equation*}

In all the three cases the obtained strategy profiles are valid, so we have a Nash equilibrium. In the first case the equilibrium is interior.
Notice that in all the three cases, the ratios of the probabilities for battlefield $k\in B$, $p_1^k/p_2^k$, are equal to the inverse ratio of the corresponding quantities $\mu_2^k/\mu_1^k$.
The equilibrium and the relative probability of winning a node are illustrated in Figure \ref{figureequilibriumtwohubs}. The equilibrium efforts on own hub are increasing and on the other nodes are decreasing in $\lambda$. When $\lambda$ reached the threshold $L_1$, a player allocates zero direct effort to the hub node of the other player. When $\lambda$ reaches $L_2$, a player allocates zero to all spoke nodes. The relative probability of winning the spokes is equal to $1$ for all $\lambda$, and this is intuitive. The probability of winning own hub and other's hub exhibits an interesting non-monotonicity: when $\lambda$ increases from a low value, the relative probability of winning own hub (other's hub) increases (decreases). However, as $\lambda$ continues to grow, spillovers grow and differences in effective efforts shrink, and this leads to a fall in the relative probability of winning the own hub.
$\hfill{\Box}$

\begin{figure}
\centering
\includegraphics[scale=0.80]{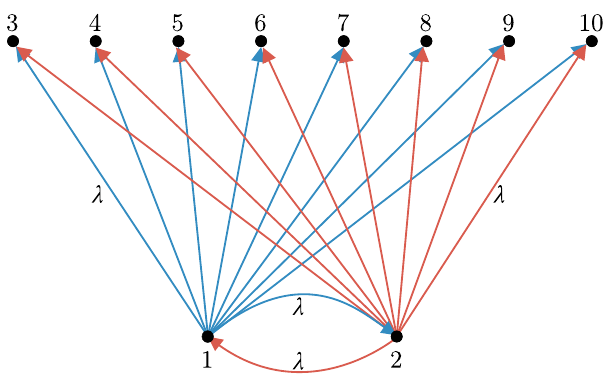}
\caption{Competing hub-spoke networks(blue: player 1; red: player 2)}
\label{figuretwohubs}
\end{figure}

\begin{figure}%
    \centering
    \subfloat{{\includegraphics[scale=0.80]{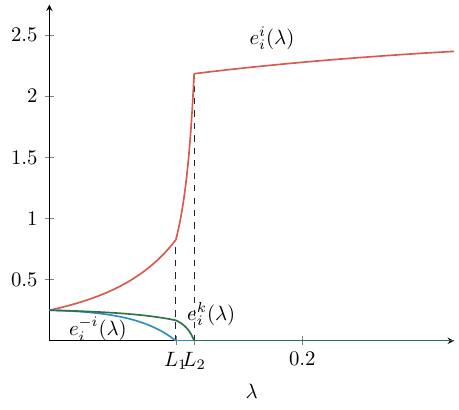}}}%
    \qquad
    \subfloat{{\includegraphics[scale=0.80]{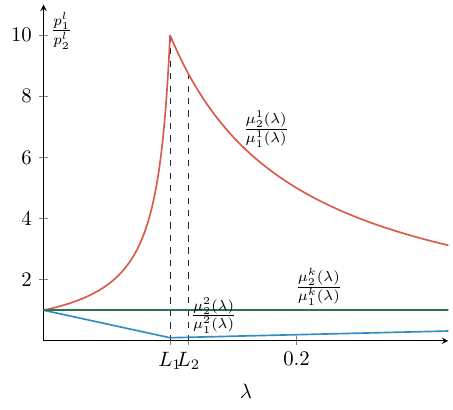}}}%
    \caption{Equilibrium efforts and Relative Prob. of winning}
    \label{figureequilibriumtwohubs}
\end{figure}

\bigskip

\section{Network Design}
\label{sec:spilldes}

So far, in our analysis, we have taken the spillovers as exogenous. Our analysis suggests that networks of spillovers shape the allocation of efforts on nodes and the payoffs of the players. Inspired by the large literature on contest design, we now ask how networks can be designed to attain specific objectives such as maximizing aggregate total efforts of players and maximizing the total utility of players. 

\subsection{Maximizing efforts}

Our characterization of equilibrium in Theorem~\ref{th:char} shows that individual efforts at specific nodes and that the aggregate efforts depend on the network of spillovers. In the  literature, a major question pertains to the design of contests that maximize efforts (see e.g., \cite{Moldovanu2001optimal},  \cite{fu2020optimal}, \cite{hinnosaar2024optimal}. The novelty in our paper is that we use the (directed weighted) network of spillovers to structure incentives in contests. 

We will consider the following sequence of moves. The designer, taking as given the costs and value of the battles, first chooses the network of spillovers and observing this network, the two players then simultaneously choose efforts at different battles.

We start by noting an implication of the characterization result. Recall that the relative probability of a player $i$ winning a node $k$ is
\begin{equation*}
\frac{p^k_i}{p^k_j} = \left(\frac{\mu^k_{j} c_{j}}{\mu^k_{i} c_{i}}\right)^{\gamma}.
\end{equation*}
So, if the two competing networks are the same, then the $\mu$ terms drop out. This means that the probability of winning will depend only on the costs of effort.  Since the equilibrium total efforts of the players as well as equilibrium payoffs (c.f. Equations~\eqref{eq:char:total} and~\eqref{eq:char:payoff} of Theorem~\ref{th:char})  are a function of these probabilities and the costs, they are also independent of the common network of spillovers. 

This leads to a type of `network invariance' result: in equilibrium, the probabilities of winning the battlefields, the total efforts, and the payoffs of the players are independent of the network of spillovers, if the two competing networks are equal.

\begin{proposition}
Suppose that the networks of both players are the same. In equilibrium, the aggregate efforts and payoffs of the players are invariant with respect to the network.
\end{proposition}

The proof of the result follows directly from our observations above and is omitted. This result tells us that if networks are to matter for aggregate efforts and for payoffs, then they must be distinct. To demonstrate how the choice of the networks of spillovers can be used to maximize the equilibrium total efforts, consider our two-node example above (see Example \ref{exampletwonodes}). Figure \ref{fig:effortstwonodes} shows that aggregate total efforts across both players increase until a certain point and then decrease. Recall that, in this example, the spillovers only apply to player 2, and this is the player who has higher costs of effort. Therefore, a larger $\lambda$ may be seen as a form of handicapping: this enables the weaker player to catch up with the stronger player. Also note that beyond a certain point the spillovers lower aggregate total efforts. So there is a well defined level of spillover at which these efforts are maximized. We build on this example in our study of networks that maximize aggregate total efforts of players.

\begin{figure}
\centering
\includegraphics[scale=.85]{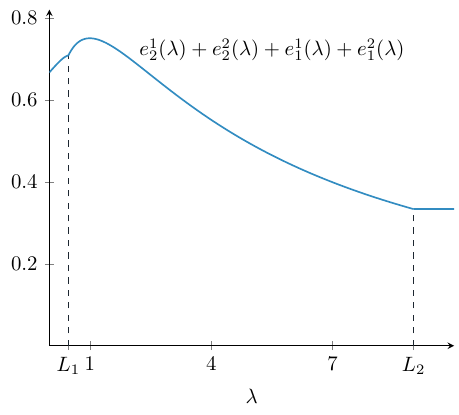}
\caption{Aggregate Efforts and Spillovers}
\label{fig:effortstwonodes}
\end{figure}

Theorem~\ref{th:char1} tells us that equilibrium total efforts of players:
\begin{equation*}
\sum_{k \in B} e^k_i = \frac{\gamma}{c_i} \left(\sum_{k\in B} p_1^k p_2^k v^k\right) \textrm{, for all $i \in \{1,2\}$}.
\end{equation*}
So to maximize effort, we need to maximize the product of the equilibrium-winning probabilities. Since $p^k_2 = 1 - p^k_1$, we get $p^k_1 = p^k_2 = 1/2$. We  design the spillovers to ensure that this holds. In particular, if
\begin{equation*}
\mu^k_1 = \frac{c_2}{c_1} \mu^k_2
\end{equation*}
then the equilibrium winning probabilities are equal to $1/2$. We build on this argument to establish the following result. Recall that $m$ is the number of battlefields. 

\begin{proposition}
\label{pr:maxeffort}
Suppose $c_2 \geq c_1 > 0$ and $v^k=v > 0$ for all $k\in B$. The maximal total efforts of players are attained when network of player 1 is empty and network for player 2 is a complete directed network with
\begin{equation*}
\rho^{l,k}_2 = \left(\frac{1}{m-1}\right) \left(\frac{c_2 - c_1}{c_1}\right) \textrm{, for all $l\in B$ and $k \in B\setminus\{l\}$}.
\end{equation*}
The \textit{handicap} offered to player $2$ is
\begin{equation*}
\bm{1}^{T}(\bm{\rho}_2 - \bm{\rho}_1) \bm{1} = m\left(\frac{c_2 - c_1}{c_1}\right).
\end{equation*}
\end{proposition}

\begin{proof}[Proof:]
Assume that $c_2 \geq c_1 > 0$ and $v^1 = \ldots = v^m = v > 0$. 
By Theorem~\ref{th:char}, the equilibrium total effort of player $i$ is equal to
\begin{equation*}
\sum_{k \in B} e_i^k = \frac{\gamma}{c_i}\sum_{k \in B} p_1^k p_2^k v^k = \frac{\gamma v}{c_i}\sum_{k \in B} p_1^k p_2^k.
\end{equation*}
Since, for all $k \in B$, $p^k_1 +p^k_2 = 1$, $p^k_1 \geq 0$, and $p^k_2 \geq 0$, so, keeping $\gamma$, $c_i$, and $v$ constant, the total effort is maximised when $p^k_1 = p^k_2 = 1/2$. In this case the total effort is equal to
\begin{equation*}
\sum_{k \in B} e_i^k = \frac{m \gamma v}{4c_i}
\end{equation*}
and this is the upper bound on the equilibrium total effort. To see that this bound is attainable, let $\bm{\rho}_1$ be the empty network with 
\begin{equation*}
\rho^{l,k}_1 = 0 \textrm{, for all $l\in B$ and $k \in B\setminus\{l\}$},
\end{equation*}
and $\bm{\rho}_2$ be complete directed networks with
\begin{equation*}
\rho^{l,k}_2 = \frac{c_2 - c_1}{(m-1)c_1} \textrm{, for all $l\in B$ and $k \in B\setminus\{l\}$}.
\end{equation*}
Notice that
\begin{equation*}
\bm{\mu}_1 = \bm{1} \quad \textrm{and} \quad \bm{\mu}_2 = \frac{c_1}{c_2} \bm{1} = \frac{c_1}{c_2} \bm{\mu}_1
\end{equation*}
solve equations $(\mathbf{I} + \bm{\rho}_i) \bm{\mu}_i = \bm{1}$, for $i \in \{1,2\}$, respectively.
In addition, using~\eqref{eq:char:pr} in Theorem~\ref{th:char}, for all $k \in B$,
\begin{equation*}
p_1^k p_2^k = \frac{(\mu_{1}^k \mu_{2}^k c_{1}c_{2})^\gamma}{\left(\left(\mu_{1}^k c_{1}\right)^{\gamma}+ \left(\mu_{2}^k c_{2}\right)^{\gamma}\right)^2} = \frac{1}{4}.
\end{equation*}

To complete the proof for the upper bound on total efforts, we need to show that the game has an equilibrium with the values of $\bm{\mu}_1$ and $\bm{\mu}_2$ given above. Notice that
\begin{equation*}
\left(\mathbf{I}+\bm{\rho}_1\right)^{-1} = \mathbf{I}\qquad \textrm{and} \qquad \left(\mathbf{I}+\bm{\rho}_2\right)^{-1} = \frac{c_1(c_2 - c_1)}{c_2(c_2 - m c_1)} \bm{1} - \frac{(m-1)c_1}{c_2-m c_1} \mathbf{I}.
\end{equation*}
(where $\bm{1}$ denotes a $m\times m$ matrix of $1$'s).
Using~\eqref{eq:char:e} in Theorem~\ref{th:char},
\begin{equation*}
\bm{e}_i = \frac{\gamma}{c_i} \left(\mathbf{I} + {\bm{\rho}_i}^T\right)^{-1} \left(\bm{p}_1 \odot\bm{p}_2 \odot \bm{v} \oslash \bm{\mu}_i\right)  = \frac{\gamma}{4c_i} \left(\mathbf{I} + {\bm{\rho}_i}^T\right)^{-1}\left(\bm{v} \oslash \bm{\mu}_i\right).
\end{equation*}
Hence, for $k \in B$,
\begin{equation*}
e^{k}_i = \frac{\gamma v}{4c_i}.
\end{equation*}
Since all the efforts are greater than zero, the game has an interior equilibrium $(\bm{e}_1,\bm{e}_2)$ with the associated marginal rates of substitution $(\bm{\mu}_1,\bm{\mu}_2)$, as computed above, and the winning probabilities equal to $1/2$ for both players, across all battlefields. By Lemma~\ref{lemma:unique} in the appendix and non-singularity of matrices $\mathbf{I} + \bm{\rho}_i$, for all $i \in \{1,2\}$, this constructed equilibrium is the unique one.
Thus we established the upper bound on the total effort and the networks of spillovers that allow us to attain it. The handicap offered by the social planner to player 2 is equal to
\begin{equation*}
\bm{1}^{T} (\bm{\rho}_2 - \bm{\rho}_1)\bm{1} = \bm{1}^{T} \bm{\rho}_2 \bm{1} = m\left(\frac{c_2- c_1}{c_1}\right).
\end{equation*}
This completes the proof.
\end{proof}

We note that the networks of spillovers maximizing the total efforts of the players in Proposition~\ref{pr:maxeffort} are not unique. In fact, there is a continuum of such networks. The designer could start with a complete network of spillovers $\bm{\rho}_1$ where the spillover between any two battlefields is set to the same value $\lambda_1 \geq 0$. The associated network $\bm{\rho}_2$ is also a complete network of spillovers with the spillovers between any two battlefields set to the same value
\begin{equation*}
\lambda_2 = \frac{c_2}{c_1}\lambda_1 + \left(\frac{1}{m-1}\right) \left(\frac{c_2 - c_1}{c_1}\right).
\end{equation*}
The handicap of player $2$ is in this case equal to
\begin{equation*}
\bm{1}^{T}(\bm{\rho}_2 - \bm{\rho}_1) \bm{1} = m\left(\frac{c_2 - c_1}{c_1}\right)(1 + (m-1)\lambda_1).
\end{equation*}
The construction in  Proposition~\ref{pr:maxeffort} corresponds to  $\lambda_1=0$.

The above result was established for the case where all battles have equal values. When prizes in different battles are very unequal, there may sometimes not exist interior equilibrium with the network design we presented above. In this case, network design takes on a slightly more complicated form and this is studied in the online appendix.    

\subsection{Maximizing Payoffs} 

We next turn to design of networks that maximize payoffs of players. By Theorem~\ref{th:char}, the equilibrium social welfare is
\begin{equation}
\payoff{1}(\bm{e}_1,\bm{e}_2) + \payoff{2}(\bm{e}_1,\bm{e}_2) = \sum_{k \in B} \left(1 - 2\gamma p^k_1 p^k_2\right) v^k
= \sum_{k \in B} v^k - \frac{2c_1 c_2}{c_1 + c_2} \sum_{k\in B}\left(e^k_1 + e^k_2\right).
\end{equation}

Hence social welfare maximization is equivalent to minimization of the sum of total efforts of the players. By Theorem~\ref{th:char}, on any network the equilibrium total effort of each player is greater than zero. This means that the upper bound on equilibrium social welfare is equal to the sum of the value of the battlefields. When the values of the battlefields are equal, by choosing an empty network for one player and a complete network of spillovers with sufficiently large magnitude of spillovers for the other player, the social planner can attain equilibrium social welfare that is arbitrarily close to this upper bound.

\begin{proposition}
\label{pr:maxwelfare}
Let $p$ be the Tullock CSF with $\gamma \in (0,1]$ and let $c_1 > 0$ and $c_2 > 0$ be the costs of effort of the players.
Suppose that the values of all battlefields are equal, $v^1 = \ldots = v^m = v > 0$.

For any $\varepsilon > 0$ there exists $\lambda_{\varepsilon} > 0$ such that the equilibrium total effort of the two players attained when $\bm{\rho}_1$ is an empty network with
\begin{equation*}
\rho^{l,k}_1 = 0 \textrm{, for all $l\in B$ and $k \in B\setminus\{l\}$}
\end{equation*}
and $\bm{\rho}_2$ is a complete directed network with
\begin{equation*}
\rho^{l,k}_2 = \lambda_{\varepsilon} \textrm{, for all $l\in B$ and $k \in B\setminus\{l\}$},
\end{equation*}
is less than $\varepsilon$.
\end{proposition}

\begin{proof}
For the lower bound, notice that keeping $c_i$ and $v$ constant, the expression in
\begin{equation}
\label{eq:maxmineffort:1}
\sum_{k \in B} e_i^k = \frac{\gamma}{c_i}\sum_{k \in B} p_1^k p_2^k v^k
\end{equation}
is minimized when either $p^k_1 = 0$ and $p^k_2 = 1$ or $p^k_1 = 1$ and $p^k_2 = 0$. In this case the total effort is $0$. By Theorem~\ref{th:char}, every battlefield receives positive effective effort in equilibrium, which implies that in equilibrium the lower bound of $0$ is not attained for any pair of spillover networks $\bm{\rho}_1$ and $\bm{\rho}_2$.

Let $\bm{\rho}_1$ be an empty network
\begin{equation*}
\rho^{l,k}_1 = 0 \textrm{, for all $l\in B$ and $k \in B\setminus\{l\}$}
\end{equation*}
and let $\bm{\rho}_2$ be complete directed networks with
\begin{equation*}
\rho^{l,k}_2 = \lambda \textrm{, for all $l\in B$ and $k \in B\setminus\{l\}$},
\end{equation*}
with $\lambda > 0$.
Notice that $\bm{\mu}_i = M_i \bm{1}$ where
\begin{equation*}
M_1 = 1 \quad \textrm{and} \quad M_2 = \frac{1}{1+(m-1)\lambda}
\end{equation*}
solve equations $(\mathbf{I} + \bm{\rho}_i) \bm{\mu}_i = \bm{1}$, for $i \in \{1,2\}$, respectively.
Moreover, using~\eqref{eq:char:pr} in Theorem~\ref{th:char}, for all $k \in B$,
\begin{align*}
p_1^k p_2^k & = \frac{(\mu_{1}^k \mu_{2}^k c_{1}c_{2})^\gamma}{\left(\left(\mu_{1}^k c_{1}\right)^{\gamma}+ \left(\mu_{2}^k c_{2}\right)^{\gamma}\right)^2} = \frac{1}{\left(\frac{M_1 c_1}{M_2 c_2}\right)^{\gamma} + \left(\frac{M_2 c_2}{M_1 c_1}\right)^{\gamma} + 2} \\
& = \frac{1}{\left(\frac{(1+(m-1)\lambda) c_1}{c_2}\right)^{\gamma} + \left(\frac{c_2}{(1+(m-1)\lambda) c_1}\right)^{\gamma} + 2} := Q(\lambda).
\end{align*}
Hence, for all $k\in B$ and $i\in \{1,2\}$,
\begin{equation*}
\lim_{\lambda \rightarrow +\infty} p_1^k p_2^k = 0
\end{equation*}
and, consequently,
\begin{equation*}
\lim_{\lambda \rightarrow +\infty} \frac{\gamma}{c_i}\sum_{k \in B} p_1^k p_2^k v^k = 0.
\end{equation*}

To complete the proof we need to show that the game has an equilibrium with the values of $\bm{\mu}_1$ and $\bm{\mu}_2$ given above. Arguing as in the proof of the maximizing effort design result above, notice that
\begin{equation*}
\left(\mathbf{I}+\bm{\rho}_1\right)^{-1} = \mathbf{I}\qquad \textrm{and} \qquad \left(\mathbf{I}+\bm{\rho}_2\right)^{-1} = \frac{1}{(\lambda - 1)(1 + (m-1)\lambda)} \bm{1} - \frac{1}{\lambda - 1} \mathbf{I}.
\end{equation*}

Using~\eqref{eq:char:e} in Theorem~\ref{th:char},
\begin{equation*}
\bm{e}_i = \frac{\gamma}{c_i} \left(\mathbf{I} + {\bm{\rho}_i}^T\right)^{-1} \left(\bm{p}_1 \odot\bm{p}_2 \odot \bm{v} \oslash \bm{\mu}_i\right)  = \frac{\gamma v}{c_i}\frac{Q(\lambda)}{M_i} \left(\mathbf{I} + {\bm{\rho}_i}^T\right)^{-1} \bm{1} = \frac{\gamma v}{c_i}Q(\lambda) \bm{1}.
\end{equation*}
Hence, for $k \in B$,
\begin{equation*}
e^{k}_i = \frac{\gamma v}{c_i} p_1^k p_2^k .
\end{equation*}
Since all the efforts are greater than zero, the game has an interior equilibrium $(\bm{e}_1,\bm{e}_2)$ with the associated marginal rates of substitution $(\bm{\mu}_1,\bm{\mu}_2)$, as computed above. 
By Lemma~\ref{lemma:unique} in the appendix and non-singularity of matrices $\mathbf{I} + \bm{\rho}_i$, for all $i \in \{1,2\}$, this constructed equilibrium is unique. This completes the proof.
\end{proof}

\section{Endogenous spillovers}
\label{sec:endogenous}

The analysis of equilibrium in given networks illustrates the powerful effects of network spillovers in shaping strategies and payoffs. It is then natural to ask what would be the outcome if players could invest in both the efforts and create their own network of spillovers. As costs of links are zero (for consistency with the network design section), to make the problem interesting, we need to assume a bound on the  spillovers between the nodes. For simplicity, we set the bound to be equal to 1.  Let $\mathcal{G}(B) = \{\bm{\rho} \in [0,1]^{B \times B} : \textrm{for all $k \in B$, } \rho^{k,k} = 0 \}$ be the set of all admissible networks of spillovers over the set of battlefields $B$.
Each player $i \in \{1,2\}$ chooses the efforts at each battlefield, $\bm{e}_i \in \mathbb{R}^{B}_{\geq 0}$, as well as the magnitudes of spillovers between the battlefields, $\bm{\rho}_i \in \mathcal{G}(B)$. The spillovers are costless, in line with the section on network design. The choices are made simultaneously and player $i$'s payoff from strategy profile $(\bm{e}_1,\bm{e}_2,\bm{\rho}_1,\bm{\rho}_2)$
is
\begin{align*}
\payoff{i}(\bm{e}_1,\bm{e}_2,\bm{\rho}_1,\bm{\rho}_2) & = \sum_{k \in B} v^k p_i^k(\bm{e}_1,\bm{e}_2,\bm{\rho}_1,\bm{\rho}_2) - c_i \sum_{k \in B} e_i^k  \\
{} & = \sum_{k \in B} v^k p_i^k(\bm{e}_1,\bm{e}_2,\bm{\rho}_1,\bm{\rho}_2) - c_i \bm{1}^{T} \bm{e}_i,
\end{align*}
where
\begin{equation*}
p_i^k(\bm{e}_1,\bm{e}_2,\bm{\rho}_1,\bm{\rho}_2) = \frac{\left(y^k_i\right)^{\gamma}}{\left(y^k_1\right)^{\gamma} + \left(y^k_2\right)^{\gamma}}
\end{equation*}
and $\bm{y}_j = (\mathbf{I} + \bm{\rho}_j^T) \bm{e}_j$, for all $j \in \{1,2\}$.

\begin{proposition}   
\label{prop-zerolinkcost}
A strategy profile $(\bm{e}_1,\bm{e}_2,\bm{\rho}_1,\bm{\rho}_2)$ is an equilibrium in the game with costless spillovers if and only if, for any $i \in \{1,2\}$, $\bm{\rho}_i$ contains an out-star with spillovers of magnitude $1$ and, for all $k \in B$, the effective effort at each battlefield is equal to the total effort,
\begin{equation*}
y_i^k = \sum_{l \in B} e_i^l = \frac{\gamma}{c_i} \, \frac{c_{1}^\gamma c_2^\gamma}{(c_1^\gamma + c_2^\gamma)^2} \left(\sum_{k \in B} v^k\right).
\end{equation*}
In every equilibrium, payoff to player $i$ equals
\begin{equation*}
\payoff{i}=\left(\frac{c_{-i}^\gamma}{c_1^\gamma+c_2^\gamma} - \gamma \, \frac{c_{1}^\gamma}{c_1^\gamma + c_2^\gamma} \frac{c_{2}^\gamma}{c_1^\gamma + c_2^\gamma}\right) \left(\sum_{k \in B} v^k\right).
\end{equation*}
\end{proposition}

\begin{proof} 
Let $(\bm{e}_1,\bm{e}_2,\bm{\rho}_1,\bm{\rho}_2)$ be a Nash equilibrium of the game. 
We show first that for each player $i \in \{1,2\}$ and each battlefield $k \in B$, $y_i^k > 0$. For assume otherwise and let $k \in B$ be such that $y_i^k = 0$ for some player $i \in N$. By Theorem~\ref{th:char}, $y_{-i}^k > 0$. Moreover, there exists a battlefield $l$ such that $e_i^l > 0$. Since $y_i^k = 0$ so $l\neq k$ and $\rho_i^{lk} = 0$. Since $y_{-i}^{k} > 0$, player $i$ strictly benefits from increasing $\rho_i^{lk}$. A contradiction with the strategy profile being a Nash equilibrium.
Second, we show that, for any $i \in \{1,2\}$, $\bm{\rho}_i$ contains an out-star with outgoing spillovers of magnitude $1$. For assume otherwise. Take the battlefield $k \in B$ such that $e_i^k = \max_{l \in B} e_i^l$. There exists a battlefield $l$ such that $\rho_i^{kl} < 1$. Since, as we showed above, $y_{-i}^l > 0$ so player $i$ strictly benefits from increasing $\rho_i^{kl}$. Again, a contradiction with the strategy profile being a Nash equilibrium.
Third, we show that for any player $i \in \{1,2\}$ and every battlefield $k$ the effective effort at $k$ is equal to the total effort of $i$,
$y_i^k = E_i = \sum_{l \in B} e_i^l$.
For assume otherwise and let $l$ be a battlefield with $y_i^l < E_i$. Since all the spillovers are bounded from above by $1$, an effective effort at every battlefield is less than or equal to $E_i$. Let $k$ be a centre of an out-star with outgoing spillovers of magnitude $1$ in $\rho^k_i$. Moving all the effort to $k$ result in all battlefields having effective effort equal to $E_i$. Since $y_i^l < E_i$ and, as we showed above, all battlefields receive positive effective effort from the other player, this strictly increases the payoff of player $i$. A contradiction with the strategy profile being a Nash equilibrium.

The aggregate efforts $(E_1,E_2)$ must be the equilibrium profile of a two-player one-battle contest where the prize equals  $V=\left(\sum_{k\in B} v^k\right)$.  By the standard results in the literature, for $i \in \{1,2\}$,
\begin{equation*}
E_i = \frac{\gamma}{c_i}p_i(1-p_i) V \qquad \textrm{ where } \qquad p_i=\frac{c_{-i}^\gamma}{c_1^\gamma+c_2^\gamma}.
\end{equation*}
\end{proof}

 Every equilibrium of the game exhibits a \textit{universal access} property, the efforts at any battlefield are made available to all battlefields.  Given the analysis above, it is easy to see that a strategy profile $(\bm{e}_1,\bm{e}_2,\bm{\rho}_1,\bm{\rho}_2)$, where $\bm{\rho}_i$  is an out-star with center $k\in B$, $e_i^k = E_i$, and $e_i^l = 0$ for $l \in B\setminus \{k\}$, is an Nash equilibrium. 
 An example of an equilibrium is presented in Figure~\ref{fig:endogenous1}.

 \begin{figure}[htp]
\begin{center}
  \includegraphics[scale=0.6]{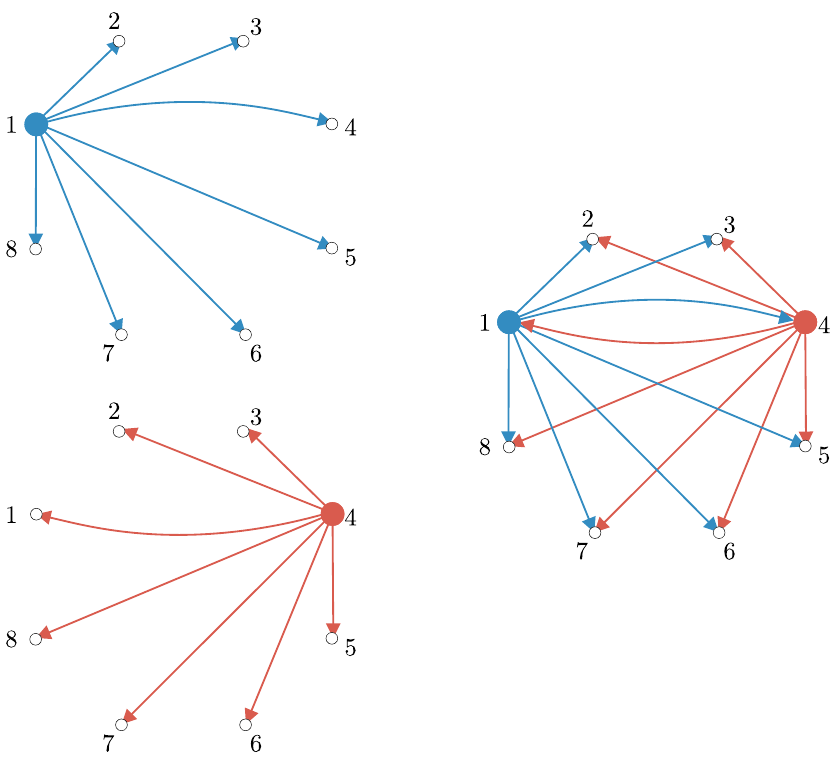}
\end{center}
  \caption{Equilibrium in networks and efforts ($n=8$). Left panel -- individual strategies;  right panel: equilibrium outcome.}
  \label{fig:endogenous1}
\end{figure}

In network design section, we assumed that spillovers were unbounded, whereas in this section we assume that spillovers are bounded. So a direct comparison between networks that maximize aggregate efforts (and those that maximize aggregate utility) with those arise when players choose their own networks is not possible. However, it is easy to see that they will typically be quite different. Networks that maximize aggregate efforts will aim for unequal spillovers between the players to reduce the difference in cost of efforts; by contrast, the decentralized choice of networks results in the two players choosing maximal possible spillovers.

\section{Concluding Remarks}

We studied a model of multi-battle contests with efforts spillover between the battles. We established existence and uniqueness of equilibrium. The equilibrium probability of a player winning a battlefield can be expressed as a product of the ratio of costs of two players and the ratio of the centrality of the battlefields in the two networks, respectively. This result brings together the central insight of the research on networks and contests. We studied the design of networks that maximize efforts: here we showed that networks can serve the role of \textit{handicapping} the lower cost player. Finally, we studied the choice of networks and efforts in battles by players: we showed that a player's efforts on any battlefield are available to every battlefield. This \textit{universal access} property property implies that players will choose efforts in proportion to the sum of the value of all battlefields and the relative costs of efforts of the two players. 

Our analysis is conducted for two-player games, but we expect that the main arguments should carry over to general $n$ player games. Similarly, we believe that our arguments can be extended to cover more general costs of effort (an earlier version of the paper derived similar results in a model with exogenously specified constant resources for each of the players).  

In future work it would be interesting to study alternative ways of aggregating winnings from battlefields onto payoffs -- such as e.g., the player who wins a majority of battles earns positive payoff and the other player earns zero payoffs. Finally, our formulation of spillovers across nodes reflects the idea that efforts on one node are substitutes for efforts on neighbouring nodes. It would be interesting to examine the case where they are complements.

\newpage
\bibliographystyle{plainnat}
\bibliography{biblio}

\newpage

\appendix
\section{Appendix}

\subsection{Equilibrium: existence and uniqueness}

The first step is to establish existence. 
\begin{lemma}
\label{lemma:existence}
If $p$ is the Tullock CSF with $\gamma \in (0,1]$, then there exists a pure strategy Nash equilibrium.
\end{lemma}

\begin{proof}
Let $\Gammait$ denote the original conflict game. For any positive $\varepsilon>0$ we define a truncated conflict game  $\Gammait^{\varepsilon}$ in which each player's effort in any battle is bounded below by  $\varepsilon$.
In the truncated game $\Gammait^{\varepsilon}$, each $i$'s payoff, as defined in~\eqref{eq:payoff}, is continuously differentiable.  Moreover, for any $i\in \{1,2\}$, $\payoff{i}$  is concave in  $\bm{e}_i$,  as  $\bm{y}_i$ results from an affine transformation of $\bm{e}_i$, the CSF $p_i$ is concave in $y_i^k$ (as $p_i$ is the Tullock CSF with $\gamma \in (0,1]$),
and the cost is convex in $\bm{e}_i$. In addition, since the total prize is bounded from above by $\sum_{k\in B} v^k$  and the winning probabilities at every battlefield are bounded from above by $1$, there is an upper bound $M$ on the efforts so that exerting an effort higher than $M$ on any battlefield is strictly dominated. Hence, without loss of generality, we can assume that each player's strategy space is $[\varepsilon,M]^B$.
If $\varepsilon < M$ then each player's strategy space is convex, compact, and non empty. Therefore, for any $\varepsilon\in (0,M)$, by Glicksberg's fixed point theorem~(\cite{Glicksberg1952}), a pure strategy Nash equilibrium of $\Gammait^{\varepsilon}$ exists

For each positive integer value $l = 1,2,\ldots$ let $\bm{e}^*(l) = (\bm{e}_1^*(l),\bm{e}_2^*(l))$ denote an equilibrium of $\Gamma^{1/(q+l)}$ where
$q > 1/M$ is an integer. Since the sequence $(\bm{e}^*(l))_{l=1}^{+\infty}$ lies in the compact set $[0,M]^B\times [0,M]^B$, it has a convergent subsequence. To simplify the notation, we may assume that the sequence $\bm{e}^*(l)$ itself converges to a limit $\bm{e}^*(+\infty)$.

Let $\bm{y}_i^*(l) = (\mathbf{I} + \bm{\rho}_i^T)\bm{e}^{*}_i(l)$ be the effective efforts of player $i \in \{1,2\}$ associated with the true efforts $\bm{e}^*(l)$, for $l\in\mathbb{N}\cup\{+\infty\}$.
We will show that $\bm{e}^*(+\infty)$ is an equilibrium of the original game. To this end,  we will show that:
\begin{enumerate}[(i)]
\item
for any battle $k\in B$ there exists a player $i \in \{1,2\}$ such that $y^{*k}_i(+\infty)>0$, and \label{p:existence:1}
\item
$\bm{e}^*(+\infty)$ is indeed an equilibrium.\label{p:existence:2}
\end{enumerate}

For point~\eqref{p:existence:1} assume, to the contrary, that there exists a battlefield $k\in B$ such that
for all $i \in \{1,2\}$, $y_i^{k^{*}}(+\infty)=0$.

Let
\begin{equation*}
p_i^{k^{*}}(l) = \frac{(y_i^{k^{*}}(l))^{\gamma}}{(y_1^{k^{*}}(l))^{\gamma} + (y_2^{k^{*}}(l))^{\gamma}}
\end{equation*}
be player $i$'s winning probability at battlefield $k$ in $\Gamma^{1/(q+l)}$ under strategy profile $\bm{e}^{*}(l)$.
Since $p_1^{k^{*}}(l) + p_2^{k^{*}}(l)=1$, there exists a player $j \in \{1,2\}$ such that $p_j^{k^{*}}(l) \leq 1/2$
for infinitely many $l \in\mathbb{N}$. Taking a subsequence if necessary, we can assume that $p_j^{k^{*}}(l)\leq 1/2$ for any $l\in\mathbb{N}$.

Now consider
\begin{align*}
\left.\frac{\partial \payoff{j}(\bm{e})}{\partial e_j^k}\right|_{\bm{e}=\bm{e}^*(l)}
& \geq -c_j+  v^k\left.\frac{\partial p_j^k(\bm{e})}{\partial e_j^k}\right|_{\bm{e}=\bm{e}^*(l)}\\
& = -c_j+  v^k \left(1-\frac{(y_j^{k^{*}}(l))^{\gamma}}{(y_1^{k^{*}}(l))^{\gamma} + (y_2^{k^{*}}(l))^{\gamma}} \right) \frac{1}{(y_1^{k^{*}}(l))^{\gamma} + (y_2^{k^{*}}(l))^{\gamma}}
\\
& \geq -c_j+  v^k \left(1-\frac{1}{2}\right)\frac{1}{(y_1^{k^{*}}(l))^{\gamma} + (y_2^{k^{*}}(l))^{\gamma}},
\end{align*}
where, in the first step, we ignore the nonnegative spillovers of $e_j^k$ on other battlefields $k \in B$, the
second step follows from direct computation, and the last step follows from the fact that $p_j^{k^{*}}(l) \leq 1/2$.
Since $y_i^{k^{*}}(+\infty)=0$, for all  $i\in N$,  we have
$\lim_{l\rightarrow +\infty} \left.\frac{\partial \payoff{j}(\e)}{\partial e_j^k}\right|_{\e=\e^*(l)}=+\infty$, as $\lim_{l\rightarrow +\infty}y_i^{k^{*}}(l)=y_i^{k^{*}}(+\infty)=0$, for all $i\in \{1,2\}$.  Note that
$\lim_{l\rightarrow +\infty} e_j^{k^{*}}(l) \leq \lim_{l\rightarrow +\infty} y_j^{k^{*}}(l)=0$, implying that $\lim_{l\rightarrow +\infty} e_j^{k^{*}}(l)=0$.

Consequently,  for sufficiently large $l$, $\left.\frac{\partial \payoff{j}(\e)}{\partial e_j^k}\right|_{\bm{e}=\bm{e}^*(l)}>0$ and $e_j^{k^{*}}(l)< M$,  implying that $j$ can strictly improve his payoff by slightly increasing his effort in battle $k$, which contradicts the fact that $\bm{e}^*(l)$ is an equilibrium of $\Gammait^{1/(q+l)}$.

For point~\eqref{p:existence:2}, take any battle $k\in B$ and any player $i \in \{1,2\}$. We  need to show that
for all $\bm{e}_i \in \mathbb{R}_{\geq 0}^{B}$,
\begin{equation*}
\payoff{i}(\bm{e}_1^*(+\infty), \bm{e}^*_2(+\infty )) \geq \payoff{i}(\bm{e}_i, \bm{e}_{-i}^*(+\infty )).
\end{equation*}

Since $\bm{e}^*(l)$ is an equilibrium of $\Gammait^{1/q+l}$,  so,  for any $i\in \{1,2\}$ and any $\bm{e}_i \in [1/l,M]^B \subseteq \mathbb{R}_{\geq 0}^{B}$,
\begin{equation}
\label{eq:payoff-l}
\payoff{i}\left(\bm{e}^{*}_1(l), \bm{e}^{*}_2(l)\right) \geq \payoff{i}\left(\bm{e}_{i}, \bm{e}^{*}_{-i}(l)\right).
\end{equation}
Take any player $i\in \{1,2\}$ and any $\bm{e}_i\in \mathbb{R}_{\geq 0}^B$,
there are two cases to consider:
\paragraph{Case 1: $\bm{e}_i>\mathbf{0}$ (every entry is positive).}
Then, for sufficiently large $l$, $\bm{e}_i \in [1/l, M]^B$. When $l$ goes to infinity in \eqref{eq:payoff-l} we get
\begin{equation*}
\payoff{i}(\bm{e}_1^*(+\infty), \bm{e}^*_2(+\infty)) \geq \payoff{i}(\bm{e}_i, \bm{e}_{-i}^*(+\infty)),
\end{equation*}
as,  when $l \rightarrow +\infty$, $\bm{e}^*(l) \rightarrow \bm{e}^*(+\infty)$,  $\payoff{i}(\bm{e}^*_1(l), \bm{e}^*_2(l))\rightarrow \payoff{i}(\bm{e}^{*}_1(+\infty), \bm{e}^*_2(+\infty))$
(due to the continuity of $\payoff{i}$ at $\bm{e}=\bm{e}^{*}(+\infty)$ by point~\eqref{p:existence:1}),
and  $\payoff{i}(\bm{e}_1, \bm{e}^{*}_2(l)) \rightarrow \payoff{i}(\bm{e}_i, \bm{e}^{*}_{-i}(+\infty))$ (due to the continuity of $\payoff{i}$ at $\bm{e}=(\bm{e}_i, \bm{e}^{*}_{-i}(+\infty))$).\footnote{For any strategy profile $\bm{e} \in \mathbb{R}_{\geq 0}^B \times \mathbb{R}_{\geq 0}^B$,  as long as for any battlefield there is one player with strictly positive effective effort, $\payoff{i}$ is continuous at $\bm{e}$.}

\paragraph{Case 2: $\bm{e}_i \geq \mathbf{0}$.} Take any $\eta>0$, and consider $\hat{\bm{e}}_i =\eta\bm{1}+\bm{e}_i>\bm{0}$ (where $\bm{1}$ is the vector of $1$'s). Then,
\begin{equation*}
\payoff{i}(\bm{e}^{*}_1(+\infty), \bm{e}^{*}_2(+\infty )) \geq \payoff{i}(\hat{\bm{e}}_i, \bm{e}_{-i}^{*}(+\infty)),
\end{equation*}
by Case 1.

Furthermore, the winning probability of $i$ weakly increases when $i$'s efforts increases from $\bm{e}_i$ to $\hat{\bm{e}}_i$.    The  cost difference between $\hat{\bm{e}}_i$ and $\bm{e}_i$ is $c_i m\eta$. Therefore,
\begin{equation*}
\payoff{i}(\hat{\bm{e}}_i, \bm{e}^{*}_{-i}(+\infty)) \geq \payoff{i}(\bm{e}_i, \bm{e}_{-i}^{*}(+\infty)) -c_i m\eta.
\end{equation*}
Combining these inequalities yields
\begin{equation*}
\payoff{i}(\bm{e}^{*}_1(+\infty), \bm{e}^*_2(+\infty)) \geq \payoff{i}(\bm{e}_i, \bm{e}^{*}_{-i}(+\infty)) -c_i m\eta,
\end{equation*}
which holds for any $\eta>0$.
Taking $\eta \rightarrow 0^{+}$ yields $\payoff{i}(\bm{e}^{*}_1(+\infty), \bm{e}^*_2(+\infty)) \geq
\payoff{i}(\bm{e}_i, \bm{e}_{-i}^{*}(+\infty))$.
\end{proof}

Having established existence, we next take up uniqueness. In the proofs we use the standard notion of interchangeability of equilibria, defined as follows. The set of Nash equilibria is interchangeable if for any two Nash equilibria $(\bm{e}'_1,\bm{e}'_2)$ and $(\bm{e}''_1,\bm{e}''_2)$, the strategy profiles $(\bm{e}'_1,\bm{e}''_2)$ and $(\bm{e}''_1,\bm{e}'_2)$ are also Nash equilibria.

Adapting the argument of~\cite{Ewerhart2017}, we first show that Nash equilibria in our model are interchangeable.

\begin{lemma}
\label{lemma:exch}
The set of equilibria is interchangeable.
\end{lemma}

\begin{proof}
Define a new game, with the same sets of strategies and payoff of player $i \in \{1,2\}$ from strategy profile $(\bm{e}_1,\bm{e}_2)$ being
\begin{equation*}
\tilde{\payoff{i}}(\bm{e}_1,\bm{e}_2) = \payoff{i}(\bm{e}_1,\bm{e}_2) + c_{-i} \sum_{k \in B} e^k_{-i}.
\end{equation*}
Notice that the players have the same incentives in the new game and so it has the same set of Nash equilibria as the original game.
In addition, notice that 
\begin{align*}
& \tilde{\payoff{1}}(\bm{e}_1,\bm{e}_2) + \tilde{\payoff{2}}(\bm{e}_1,\bm{e}_2) = {} \\
& \qquad\qquad \payoff{1}(\bm{e}_1,\bm{e}_2) + \payoff{2}(\bm{e}_1,\bm{e}_2) + c_{1} \sum_{k \in B} e^k_{-i} + c_{1} \sum_{k \in B} e^k_{-i}
= \sum_{k \in B} v^k,
\end{align*}
as winning probabilities add up to $1$ for each battlefield $k\in B$.
Hence the new game is constant sum and so it has interchangeable equilibria. Hence equilibria in out game are interchangeable as well.
\end{proof}

Let $E_i$ be the set of equilibrium efforts of player $i \in \{1,2\}$ and let
\begin{equation*}
T_i = \left\{\left(\mathbf{I} + \bm{\rho}_i^T\right) \bm{e}_i : \bm{e}_i \in E_i\right\}
\end{equation*}
be the set of equilibrium effective efforts of player $i$.
We next show the following auxiliary lemma.

Define the equilibrium effective effort set for player $i$ as $T_i$. We have the following lemma.
\begin{lemma}
\label{lemma:aux:2}
For any $(\bm{y}^{*}_1, \bm{y}^{*}_2) \in T_1 \times T_2$, $i \in \{1,2\}$, and $k \in B$, 
\begin{enumerate}[(i)]
\item either ${y^{*}}^k_1 > 0$ or ${y^{*}}^k_2 > 0$,\label{p:aux:2:1}
\item for any $\bm{y}^{**}_{-i} \in T_{-i}$ either ${y^{*}}^k_i = 0$ or ${y^{*}}^k_{-i} = {y^{**}}^k_{-i}$.\label{p:aux:2:2}
\end{enumerate}
\end{lemma}

\begin{proof}
For point~\eqref{p:aux:2:1}, assume, to the contrary, that $({y^{*}}^k_1,{y^{*}}^k_2) = (0,0)$, for some battlefield $k \in B$, and let $(\bm{e}^{*}_1,\bm{e}^{*}_2) \in E_1 \times E_2$ be an equilibrium profile associated with $(\bm{y}^{*}_1,\bm{y}^{*}_2)$. Then $({e^{*}}^k_1,{e^{*}}^k_2) = (0,0)$ and increasing ${e^{*}}^k_1$ by an arbitrarily small $\varepsilon > 0$ allows $1$ to win $k$ with probability $1$ and gain $v^k > 0$ for an arbitrarily small cost $\varepsilon c_1$. Hence $1$ is able to deviate to a strategy that results in a positive increase in payoff, a contradiction with the assumption that $(\bm{e}^{*}_1,\bm{e}^{*}_2)$ is a Nash equilibrium.

For point~\eqref{p:aux:2:2}, without loss of generality, we show the claim of the lemma for $i = 2$. Suppose the corresponding real efforts are $(\bm{e}^{*}_1, \bm{e}^{*}_2)$ and $\bm{e}^{**}_1$. 
Using~\eqref{eq:muhdef}, 
\begin{equation*}
    \nabla_{\bm{e}_1} \payoff{1}(\bm{e}^{*}_1, \bm{e}^{*}_2)-\nabla_{\bm{e}_1} \payoff{1}(\bm{e}^{**}_1, \bm{e}^{*}_2) = 
    \gamma c_i \begin{pmatrix}
        {\mu^{*}}_1^1 - {\mu^{**}}_1^1 + \sum_{l \in B\setminus \{1\}} \rho^{1,l}_1 \left( {\mu^{*}}_1^l - {\mu^{**}}_1^l \right) \\
        \vdots\\
        {\mu^{*}}_1^m - {\mu^{**}}_1^m + \sum_{l \in B\setminus \{m\}} \rho^{m,l}_1 \left( {\mu^{*}}_1^l - {\mu^{**}}_1^l \right)
    \end{pmatrix}
\end{equation*}
and 
\begin{equation*}
\frac{\partial \mu^k_1}{\partial y^k_1} = - \frac{\left(y^k_1\right)^{\gamma-2}\left(y^k_2\right)^{\gamma}\left(2\gamma\left(y^k_1\right)^{\gamma} + (1 - \gamma)\left(\left(y^k_1\right)^{\gamma} + \left(y^k_2\right)^{\gamma}\right)\right)}{\left(\left(y^k_1\right)^{\gamma} + \left(y^k_2\right)^{\gamma}\right)^3} < 0,
\end{equation*}
for $\gamma \in (0,1]$, $y^k_1 > 0$, and $y^k_2 > 0$. In addition, for any $y^k_1 > 0$ and $y^k_2 = 0$, $\mu^k_1 = 0$.
Hence
\begin{align*}
    &\langle \nabla_{\bm{e}_1} \payoff{1}(\bm{e}^{*}_1, \bm{e}^{*}_2)-\nabla_{\bm{e}_1} \payoff{1}(\bm{e}^{**}_1, \bm{e}^{*}_2), \bm{e}^{*}_1 - \bm{e}^{**}_1 \rangle \\
    = &\sum_{k \in B} \left({\mu^{*}}_1^k - {\mu^{**}}_1^k + \sum_{l \in B\setminus \{k\}} \rho^{k,l}_1 \left( {\mu^{*}}_1^l - {\mu^{**}}_1^l \right) \right)\left({e^{*}}_1^k - {e^{**}}_1^k \right)\\
    =&\sum_{l \in B} \left( {\mu^{*}}_1^l - {\mu^{**}}_1^l \right) \left( {e^{*}}_1^l - {e^{**}}_1^l + \sum_{k \in B\setminus\{l\}}\rho^{k,l}_1 \left({e^{*}}_1^k - {e^{**}}_1^k\right) \right) \\
    =& \sum_{l \in B} \left( {\mu^{*}}_1^l - {\mu^{**}}_1^l \right) \left( {y^{*}}_1^l - {y^{**}}_1^l\right) \leq 0,
\end{align*}
with equality only if either ${y^{*}}^k_2 = 0$ or ${y^{**}}^k_1 = {y^{*}}^k_1$, for all $k\in B$.

On the other hand, by interchangeability, $(\bm{e}^{**}_1, \bm{e}^{*}_2)$ is also an equilibrium. Given $\bm{e}^{*}_2$, by the concavity of $u_1$, we have
\begin{align*}
    &\langle -\nabla_{\bm{e}_1} \payoff{1}(\bm{e}^{*}_1, \bm{e}^{*}_2), \bm{e}^{**}_1 - \bm{e}^{*}_1  \rangle \geq 0,\\
    &\langle -\nabla_{\bm{e}_1} \payoff{1}(\bm{e}^{**}_1, \bm{e}^{*}_2), \bm{e}^{*}_1 - \bm{e}^{**}_1  \rangle \geq 0,
\end{align*}
which implies
\begin{equation*}
    \langle \nabla_{\bm{e}_1} \payoff{1}(\bm{e}^{*}_1, \bm{e}^{*}_2)-\nabla_{\bm{e}_1} \payoff{1}(\bm{e}^{**}_1, \bm{e}^{*}_2), \bm{e}^{*}_1 - \bm{e}^{**}_1  \rangle \geq 0.
\end{equation*}

Thus we must have 
\begin{equation*}
    \sum_{l \in B} \left( {\mu^{*}}_1^l - {\mu^{**}}_1^l \right) \left( {y^{*}}_1^l - {y^{**}}_1^l\right) = 0.
\end{equation*}
and so either ${y^{*}}^k_2 = 0$ or ${y^{**}}^k_1 = {y^{*}}^k_1$, for all $k \in B$. 
\end{proof}

Given an equilibrium effective efforts profile $(\bm{y}_1,\bm{y}_2) \in T_1 \times T_2$ we define the following sets of battlefields
\begin{equation}
\label{eq:sets}
\begin{aligned}
B^{+}(\bm{y}_1,\bm{y}_2) & = \{k \in B : y^k_1 > 0 \textrm{ and } y^k_2 > 0 \},\\
A_i(\bm{y}_1,\bm{y}_2) & = \{k \in B : y^k_i > 0 \textrm{ and } y^k_{-i} = 0 \}.
\end{aligned}
\end{equation}
In the next lemma we establish that sets $B^{+}$, $A_1$ and $A_2$ are unique, across all equilibria.

\begin{lemma}
\label{lemma:sets}
For any $(\bm{y}^{*}_1, \bm{y}^{*}_2) \in T_1 \times T_2$ and $(\bm{y}^{**}_1, \bm{y}^{**}_2) \in T_1 \times T_2$, $B^{+}(\bm{y}^{*}_1, \bm{y}^{*}_2) = B^{+}(\bm{y}^{**}_1, \bm{y}^{**}_2)$, $A_1(\bm{y}^{*}_1, \bm{y}^{*}_2) = A_1(\bm{y}^{**}_1, \bm{y}^{**}_2)$, and $A_2(\bm{y}^{*}_1, \bm{y}^{*}_2) = A_2(\bm{y}^{**}_1, \bm{y}^{**}_2)$.
\end{lemma}

\begin{proof}
Take any $(\bm{y}^{*}_1, \bm{y}^{*}_2) \in T_1 \times T_2$ and $(\bm{y}^{**}_1, \bm{y}^{**}_2) \in T_1 \times T_2$.

Fix any $i \in \{1,2\}$. By point~\eqref{p:aux:2:2} of Lemma~\ref{lemma:aux:2}, for any $k \in A_i(\bm{y}^{*}_1, \bm{y}^{*}_2)$, ${y^{*}}^k_{-i} = {y^{**}}^k_{-i} = 0$. Hence, by point~\eqref{p:aux:2:1} of Lemma~\ref{lemma:aux:2}, ${y^{**}}^k_{i} > 0$ and, therefore, $k \in A_i(\bm{y}^{**}_1, \bm{y}^{**}_2)$. Thus $A_i(\bm{y}^{*}_1, \bm{y}^{*}_2) \subseteq A_i(\bm{y}^{**}_1, \bm{y}^{**}_2)$. By analogous argument, $A_i(\bm{y}^{**}_1, \bm{y}^{**}_2) \subseteq A_i(\bm{y}^{*}_1, \bm{y}^{*}_2)$ and so $A_i(\bm{y}^{*}_1, \bm{y}^{*}_2) = A_i(\bm{y}^{**}_1, \bm{y}^{**}_2)$.

By point~\eqref{p:aux:2:1} of Lemma~\ref{lemma:aux:2}, for any $(\bm{y}_1,\bm{y}_2) \in T_1\times T_2$ and any $i \in \{1,2\}$, $B^{+}(\bm{y}_1,\bm{y}_2) \cup A_i(\bm{y}_1,\bm{y}_2) = B$. This, together with $A_1(\bm{y}^{*}_1, \bm{y}^{*}_2) = A_1(\bm{y}^{**}_1, \bm{y}^{**}_2)$ yields $B^{+}(\bm{y}^{*}_1, \bm{y}^{*}_2) = B^{+}(\bm{y}^{**}_1, \bm{y}^{**}_2)$ and the proof is complete.
\end{proof}

Since sets $B^{+}$, $A_1$ and $A_2$ are the same across all equilibria, in the rest of the proof we omit their qualification by effective efforts.
Now we are ready to prove uniqueness of equilibrium payoffs, total efforts and probabilities of winning battlefields.

\begin{lemma}
\label{lemma:unique}
Equilibrium effective efforts on battlefields in $B^{+}$, equilibrium total efforts, probabilities of winning battlefields, and payoffs are unique. 
\end{lemma}

\begin{proof}
We first establish uniqueness of equilibrium effective efforts on battlefields in $B^{+}$. This follows directly from uniqueness of $B^{+}$ across all equilibria (Lemma~\ref{lemma:sets}) and point~\eqref{p:aux:2:2} of Lemma~\ref{lemma:aux:2}.

Since equilibrium effective efforts are unique on battlefields in $B^{+}$, equilibrium probabilities of winning these battlefields are the same across all equilibria. In addition, by the definition and uniqueness of set $A_i$, the probability of winning a battlefield in this set is equal to $1$ for player $i$ and equal to $0$ for player $-i$. This shows uniqueness of equilibrium probabilities of winning battlefields.

For uniqueness of equilibrium total efforts take any $i \in\{1,2\}$, any equilibrium $(\bm{e}^{*}_1,\bm{e}^{*}_2) \in E_1\times E_2$ and any equilibrium strategy $\bm{e}^{**}_i \in E_i$ of player $i$. By Lemma~\ref{lemma:exch},  $(\bm{e}^{**}_i,\bm{e}^{*}_{-i})$ is an equilibrium as well and, therefore,
\begin{equation*}
\payoff{i}\!\left(\bm{e}^{*}_i,\bm{e}^{*}_{-i}\right) = \payoff{i}\!\left(\bm{e}^{**}_i,\bm{e}^{*}_{-i}\right)
\end{equation*}
This, together with uniqueness of equilibrium probabilities of winning battlefields implies uniqueness of equilibrium total efforts (as cost depends only on the total effort).

Lastly, since equilibrium probabilities of winning battlefields and equilibrium total efforts are unique, for each player $i$, so equilibrium utilities of $i$ are also unique.
\end{proof}

By Lemma~\ref{lemma:unique}, equilibrium effective efforts are unique in the case when there exists an equilibrium where each battlefield receives positive effective effort from both players. In the following lemma we show that this is the case when $\gamma \in (0,1)$.

\begin{lemma}
\label{lemma:uniqeff}
If $\gamma \in (0,1)$ then $B^{+} = B$ and equilibrium effective efforts are unique.
\end{lemma}

\begin{proof}
We first show that in equilibrium every battlefield receives positive effective effort from every player $i \in \{1,2\}$.
For assume otherwise and let $(\bm{e}_1,\bm{e}_2)$ be an equilibrium that does not satisfy this property. Let $(\bm{y}_1,\bm{y}_2)$
be the corresponding effective efforts profile. Then there exists player $i \in \{1,2\}$ and battlefield $k \in B$ such that $y^k_i = 0$ and, consequently, $e^k_i = 0$. Since, as we noticed above, every equilibrium is of type $S^1$ so $y^k_{-i} > 0$. Let $D = y^k_{-i}$ and consider a strategy $\bm{e}'_i = (\bm{e}^{-k}_{i},\varepsilon)$ of player $1$ that assigns effort $\varepsilon > 0$ to battlefield $k$ and assigns the same efforts as $\bm{e}_i$ to the remaining battlefields. The difference in payoffs
\begin{equation*}
\payoff{i}(\bm{e}'_i,\bm{e}_{-i}) - \payoff{i}(\bm{e}_i,\bm{e}_{-i}) \geq \frac{\varepsilon^\gamma}{\varepsilon^\gamma + D^\gamma}v^k - \varepsilon c_i
= \varepsilon \left(\frac{\varepsilon^{\gamma-1}}{\varepsilon^\gamma + D^\gamma}v^k - c_i\right).
\end{equation*}
Since $\gamma \in (0,1)$ so
\begin{equation*}
\lim_{\varepsilon \rightarrow 0^{+}} \frac{\varepsilon^{\gamma-1}}{\varepsilon^\gamma + D^\gamma}v^k = +\infty
\end{equation*}
in particular, there exists a sufficiently small $\varepsilon > 0$ such that the difference in payoffs is positive.
Hence, there exists $\varepsilon > 0$ such that the deviation to $\bm{e}'_i$ is profitable for player $i$, a contradiction with the assumption that $(\bm{e}_1,\bm{e}_2)$ is a Nash equilibrium. Thus in equilibrium every battlefield receives positive effective effort from every player $i \in \{1,2\}$. In other words, $B^{+} = B$ and, by Lemma~\ref{lemma:unique}, equilibrium effective efforts are unique.
\end{proof}

\newpage

\section{On-line appendix}
\subsection{Two-node Example}
This section contains the calculations on equilibrium of Example~\ref{exampletwonodes} in Section~\ref{sec:eqprop}.

Given the efforts $(\bm{e}_1,\bm{e}_2)$, the effective efforts are $(\bm{y}_1,\bm{y}_2)$, where
\begin{align*}
y_1^1 & = e_1^1 \textrm{, } & y_1^2 & = e_1^2,\\
y_2^1 & = e_2^1 \textrm{, } & y_2^2 & = \lambda e_2^1 + e_2^2.
\end{align*}

Payoffs to the players from strategy profile $(\bm{e}_1, \bm{e}_2)$ are
\begin{align*}
\payoff{1}(\bm{e}_1,\bm{e}_2) & = \frac{e_1^1}{e_1^1 + e_2^1} + \frac{e_1^2}{e_1^2 + \lambda e_2^1 + e_2^2} - (e_1^1 + e_1^2)c_1,\\
\payoff{2}(\bm{e}_1,\bm{e}_2) & = \frac{e_2^1}{e_1^1 + e_2^1} + \frac{\lambda e_2^1 + e_2^2}{e_1^2 + \lambda e_2^1 + e_2^2} - (e_2^1 + e_2^2)c_2,
\end{align*}
and partial derivatives of the payoffs are
\begin{align*}
\frac{\partial \payoff{1}(\bm{e}_1,\bm{e}_2)}{\partial e_1^1} & = \frac{e_2^1}{(e_1^1 + e_2^1)^2} - c_1\\
\frac{\partial \payoff{1}(\bm{e}_1,\bm{e}_2)}{\partial e_1^2} & = \frac{\lambda e_2^1 + e_2^2}{\left(e_1^2 + \lambda e_2^1 + e_2^2\right)^2} - c_1\\
\frac{\partial \payoff{2}(\bm{e}_1,\bm{e}_2)}{\partial e_2^1} & = \frac{e_1^1}{(e_1^1 + e_2^1)^2} + \frac{\lambda e_1^2}{(e_1^2 + \lambda e_2^1 + e_2^2)^2} - c_2\\
\frac{\partial \payoff{2}(\bm{e}_1,\bm{e}_2)}{\partial e_2^2} & = \frac{e_1^2}{\left(e_1^2 + \lambda e_2^1 + e_2^2\right)^2} - c_2.
\end{align*}
Taking
\begin{equation}
\label{eq:example1:1}
\begin{aligned}
\mu_1^1 = \frac{e_2^1}{c_1\left(e_1^1 + e_2^1\right)^2},\qquad & \mu_1^2 = \frac{\lambda e_2^1 + e_2^2}{c_1\left(e_1^2 + \lambda e_2^1 + e_2^2\right)^2},\\
\mu_2^1 = \frac{e_1^1}{c_2(e_1^1 + e_2^1)^2},\qquad & \mu_2^2 = \frac{e_1^2}{c_2\left(e_1^2 + \lambda e_2^1 + e_2^2\right)^2},
\end{aligned}
\end{equation}
partial derivatives of the payoffs can be written as
\begin{align*}
\frac{\partial \payoff{1}(\bm{e}_1,\bm{e}_2)}{\partial e_1^1} & = c_1\left(\mu_1^1- 1\right), & \qquad \frac{\partial \payoff{1}(\bm{e}_1,\bm{e}_2)}{\partial e_1^2} & = c_1\left(\mu_1^2- 1\right),\\
\frac{\partial \payoff{2}(\bm{e}_1,\bm{e}_2)}{\partial e_2^1} & = c_2\left(\mu_2^1 + \lambda\mu_2^2 - 1\right), \qquad & \frac{\partial \payoff{2}(\bm{e}_1,\bm{e}_2)}{\partial e_2^2} & = c_2\left(\mu_2^2- 1\right).
\end{align*}
Quantities $\mu_i^k$ are marginal rates of substitution between the expected reward from winning the prize and the cost of effort at the individual battlefields: \begin{equation*}
\mu_i^k = \left. \frac{\partial(p_i^k(\bm{e}_1,\bm{e}_2) v^k)}{\partial e_i^k} \right/ \frac{\partial(c_i e_i^k)}{\partial e_i^k}.
\end{equation*}

If $(\bm{e}_1,\bm{e}_2)$ is a Nash equilibrium then, at every battlefield, the effective effort of at least one player has to be positive. For otherwise each player would strictly benefit from rising his effort at this battlefield by a sufficiently small, positive, amount. Hence in equilibrium $(\bm{e}_1,\bm{e}_2) \neq (\bm{0},\bm{0})$.

If $(\bm{e}_1,\bm{e}_2)$ is a Nash equilibrium then partial derivatives at $(\bm{e}_1,\bm{e}_2)$ have to be equal to $0$ for each player $i$ and at each battlefield $k$ where the player exerts positive real effort, $e_i^k > 0$, and they have to be non-positive at each battlefield where $e_i^k = 0$.

Suppose that for all $i \in \{1,2\}$ and $k \in \{0,1\}$, $e_i^k > 0$ so that $(\bm{e}_1,\bm{e}_2)$ is an interior equilibrium.
In terms of the quantities $\mu_i^k$, the first order condition can be written as
\begin{align*}
\mu_1^1 & = 1,& \qquad \mu_1^2 = 1,\\
\mu_2^1 + \lambda\mu_2^2 & = 1, & \qquad \mu_2^2 = 1,
\end{align*}
Which has a unique solution
\begin{align*}
\mu_1^1 & = 1, & \qquad \mu_1^2 = 1,\\
\mu_2^1 & = 1 - \lambda, & \qquad \mu_2^2 = 1.
\end{align*}
From this, the equilibrium efforts are
\begin{align*}
e_1^1 & = \frac{(1-\lambda)c_{2}}{(c_1 + (1 - \lambda)c_2)^2} \textrm{, } & e_1^2 & = \frac{c_{2}}{(c_1 + c_2)^2},\\
e_2^1 & = \frac{c_{1}}{(c_1 + (1 - \lambda)c_2)^2} \textrm{, } & e_2^2 & = \frac{c_{1}}{(c_1 + c_2)^2} - \frac{\lambda c_{1}}{(c_1 + (1 - \lambda)c_2)^2}.
\end{align*}
These efforts are positive if and only if $\lambda \in [0,L_1)$, where
\begin{equation*}
L_1 = \frac{\left(c_1 + c_2\right)\left(c_1 + 3c_2 - \sqrt{(c_1 + c_2)(c_1 + 5c_2)}\right)}{2c_2^2} < 1.
\end{equation*}
Notice that when $\lambda = 0$ we obtain a well-known unique equilibrium for a two battlefield model without spillovers.
When the spillover $\lambda = L_1$, $e_2^2 = 0$ and the remaining efforts are positive. 

Notice that in equilibrium it cannot be that $e_1^2 = 0$ and $e^2_2 > 0$, because in this case player $2$ would strictly benefit from deviating to $e_2^2/2$.

Suppose that $e_2^2 = 0$ and $e_1^2 > 0$. Since $e_1^2 > 0$ so in equilibrium it must be that $y_2^2 > 0$. For otherwise player $1$ would strictly benefit from deviating to $e_1^2/2$. From the first-order condition and the fact that $e_2^2 = 0$, 
\begin{align*}
\mu_1^1 & = 1, & \qquad \mu_1^2 = 1, & \\
\mu_2^1 + \lambda\mu_2^2 & = 1.
\end{align*}
Solving this, together with~\eqref{eq:example1:1}, we obtain a unique solution in non-negative real efforts:
\begin{align*}
e_1^1 & = \frac{(c_2 + \sqrt{\lambda}(\sqrt{\lambda}-1) c_{1})(1+\sqrt{\lambda})}{(c_1(1+\lambda) + c_2)^2} \textrm{, } & e_1^2 & = \frac{(c_2 - (\sqrt{\lambda}-1) c_{1})\sqrt{\lambda}(1+\sqrt{\lambda})}{(c_1(1+\lambda) + c_2)^2},\\
e_2^1 & = \frac{(1+\sqrt{\lambda})^2 c_{1}}{(c_1(1+\lambda) + c_2)^2} \textrm{, } & e_2^2 & = 0,
\end{align*}
with
\begin{equation*}
\mu_2^1 = \frac{c_2 + \sqrt{\lambda}(\sqrt{\lambda}-1)c_1}{(1+\sqrt{\lambda})c_2}, \qquad \mu_2^2 = \frac{c_2-(\sqrt{\lambda}-1)c_1}{\sqrt{\lambda}(1+\sqrt{\lambda})c_2}.
\end{equation*}
If $\lambda \geq L_1$, $\mu_2^2 \leq 1$ and so the first-order condition is satisfied. In addition, the efforts are non-negative if $\lambda \in [L_1,L_2)$, where
\begin{equation*}
L_2 = \left(\frac{c_1+c_2}{c_1}\right)^2.
\end{equation*}
When $\lambda = L_2$, $e_1^1 = 0$ while $e_1^1$ and $e_2^1$ remain positive.

Suppose that $e_1^2 = 0$ and $e_2^2 = 0$. Since $y_1^2 = e_1^2 = 0$ so $y_2^2 > 0$ and so $e_2^1 > 0$. Moreover, it must be that $e_1^1 > 0$, for otherwise player $2$ would strictly benefit from choosing $e_2^1/2$ instead of $e_2^1$. From the first-order condition and the fact that $e_1^2 = 0$, 
\begin{align*}
\mu_1^1 & = 1,\\
\mu_2^1 + \lambda\mu_2^2 & = 1, & \qquad \mu_2^2 = 0,
\end{align*}
from which we get
\begin{align*}
e_1^1 & = \frac{c_{2}}{(c_1 + c_2)^2} \textrm{, } & e_1^2 & = 0,\\
e_2^1 & = \frac{c_1}{(c_1 + c_2)^2} \textrm{, } & e_2^2 & = 0,
\end{align*}
with
\begin{equation*}
\mu_1^2 = \frac{1}{\lambda} \left(\frac{c_1+c_2}{c_1}\right)^2 = \frac{1}{\lambda}L_2.
\end{equation*}
If $\lambda \geq L_2$, $\mu_1^2 \leq 1$ and the first-order condition is satisfied. In addition, all efforts are non-negative in this case.

\subsection{Spillovers design with unequal prizes}

We consider the problem of total effort maximization by spillovers design with general values of the battlefields. The complete networks of spillovers, used in the case of all battlefields values being equal, do not guarantee interior equilibrium in the general case. 
We will show that there exist other networks of spillovers that allow the social planner to attain the maximal equilibrium total effort.

Consider networks of spillovers $\bm{\rho}^{*}_1$ and $\bm{\rho}^{*}_2$ defined as follows.
Assume that the battlefields in $B$ are numbered in the order increasing with respect to their values, so that $v^1 \leq \ldots \leq v^m$.
Let $r = \lfloor m/2 \rfloor$ so that $m = 2r + m \bmod 2$. Partition the nodes into $r$ groups:
$r-1$ groups of size $2$ and one group of size $3$ (when $m$ is odd) or of size $2$ (when $m$ is even):
$(1,2)$, $(3,4)$, $\ldots$, $(m-1,m)$ or $(m-2,m-1,m)$. 

Fix $i \in \{1,2\}$. Let the magnitudes of spillovers for network $\bm{\rho}^{*}_i$ be defined as follows. For $l \in \{1,\ldots,r\}$ such that the $l$'th group consist of two battlefields, let the spillovers between battlefields $2l-1$ and $2l$, ${\rho^{*}}_i^{2l-1,2l} = {\rho^{*}}_i^{2l,2l-1} = \lambda^{l}_{i}$. Set the spillovers between the remaining pairs of battlefields to $0$. A network $\bm{\rho}^{*}_i$ over $m = 8$ nodes is illustrated in Figure~\ref{fig:spilldeseven}.

\begin{figure}[htp]
\begin{center}
  \includegraphics{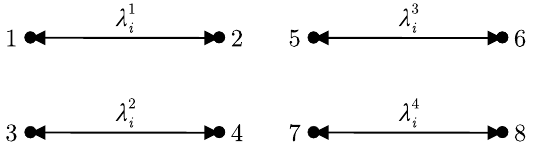}
  \caption{Network $\bm{\rho}_i^{*}$ over $m = 8$ nodes.}
  \label{fig:spilldeseven}
\end{center}
\end{figure}

In the case of $l = r$ and an odd $m$, the $l$'th group consist of three nodes, $\{m-2,m-1,m\}$. Let the spillovers from battlefield $m$ to battlefields $m-1$ and $m-2$, ${\rho^{*}}_i^{m,m-1} = {\rho^{*}}_i^{m,m-2} = \lambda^{r}_{i}/2$. 
In the case of 
\begin{equation*}
\frac{c_2}{c_1} < \frac{v^{m} + v^{m-1} + v^{m-2}}{v^{m} + v^{m-1} - v^{m-2}}
\end{equation*}
let the spillovers from battlefield $m-1$ to battlefields $m$ and $m-2$ as well as spillovers from battlefield $m-2$ to battlefields $m$ and $m-1$ be ${\rho^{*}}_i^{m-1,m} = {\rho^{*}}_i^{m-1,m-2} = {\rho^{*}}_i^{m-2,m} = {\rho^{*}}_i^{m-2,m-1} = \lambda^{r}_{i}/2$ and set all the remaining spillovers to $0$. 
In the case of 
\begin{equation*}
\frac{c_2}{c_1} \geq \frac{v^{m} + v^{m-1} + v^{m-2}}{v^{m} + v^{m-1} - v^{m-2}}
\end{equation*}
let the spillovers between battlefields $m-1$ and $m-2$ be ${\rho^{*}}_i^{m-1,m-2} = {\rho^{*}}_i^{m-2,m-1} = \lambda^{r}_{i}$ and set all the remaining spillovers to $0$.

The networks over three nodes for the two cases are presented in Figure~\ref{fig:spilldesodd}.
\begin{figure}%
    \centering
    \subfloat[\centering The case of $\frac{c_2}{c_1} < \frac{v^{m} + v^{m-1} + v^{m-2}}{v^{m} + v^{m-1} - v^{m-2}}$]{{\includegraphics{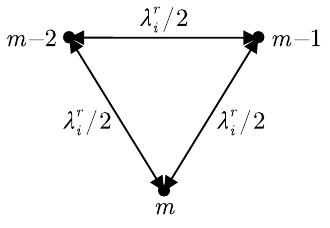} }}%
    \qquad
    \subfloat[\centering The case of $\frac{c_2}{c_1} \geq \frac{v^{m} + v^{m-1} + v^{m-2}}{v^{m} + v^{m-1} - v^{m-2}}$]{{\includegraphics{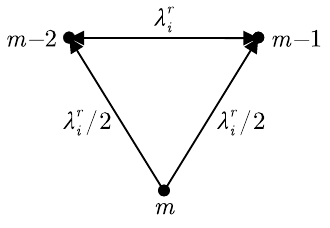} }}%
    \caption{Three node components}%
    \label{fig:spilldesodd}
\end{figure}

In the case of $i = 1$, define the magnitudes of spillovers as follows. Take any $l \in \{1,\ldots,r\}$. If the $l$'th group of nodes is of size $2$ then,
if 
\begin{equation*}
\frac{c_2}{c_1} < \frac{v^{2l-1} + v^{2l}}{v^{2l}},
\end{equation*}
take any
\begin{equation*}
\lambda^{l}_{1} \in \left[0,\frac{v^{2l-1} + v^{2l}}{v^{2l}}\frac{c_1}{c_2} - 1\right)
\end{equation*}
and if 
\begin{equation*}
\frac{c_2}{c_1} \geq \frac{v^{2l-1} + v^{2l}}{v^{2l}},
\end{equation*}
take any
\begin{equation*}
\lambda^{l}_1 > \frac{v^{2l}}{v^{2l-1}}.
\end{equation*}
If $l = r$ and group of nodes is of size $3$ then, if
\begin{equation*}
\frac{c_2}{c_1} < \frac{v^{m} + v^{m-1} + v^{m-2}}{v^{m} + v^{m-1} - v^{m-2}}
\end{equation*}
take any 
\begin{equation*}
\lambda^{l}_{1} \in \left[0,\frac{v^{m} + v^{m-1} + v^{m-2}}{v^{m} + v^{m-1} - v^{m-2}}\frac{c_1}{c_2} - 1\right)
\end{equation*}
and if
\begin{equation*}
\frac{c_2}{c_1} \geq \frac{v^{m} + v^{m-1} + v^{m-2}}{v^{m} + v^{m-1} - v^{m-2}}
\end{equation*}
take any 
\begin{equation*}
\lambda^{l}_{1} > \frac{v^{m} + v^{m-1} + v^{m-2}}{v^{m} + v^{m-1} - v^{m-2}}\frac{c_1}{c_2} - 1.
\end{equation*}

For $i = 2$, let
\begin{equation*}
\lambda^{l}_2 = \frac{c_2}{c_1} \lambda^{l}_1 + \frac{c_2 - c_1}{c_1},
\end{equation*}
for all $l \in \{1,\ldots,r\}$.
Since $c_2 \geq c_1$ and $\lambda^{l}_1 > 0$ so $\lambda^{l}_2 > 0$.

The game with the networks of spillovers $(\bm{\rho}^{*}_1,\bm{\rho}^{*}_2)$ has a unique interior equilibrium that obtains the maximum possible equilibrium total effort. 

\begin{theorem}
\label{th:maxeffort}
For any Tullock CSF $p$ with $\gamma \in (0,1]$, any costs of efforts $c_2 \geq c_1 > 0$, and any battlefield values $\bm{v}$,
and the networks of spillovers $(\bm{\rho}^{*}_1,\bm{\rho}^{*}_2)$, attain the maximal possible equilibrium total effort for each player $i \in \{1,2\}$, which is equal to
\begin{align*}
\sum_{k \in B} e^k_i = \frac{\gamma}{4c_i} \sum_{k\in B} v^k,
\end{align*}
in a unique equilibrium.
\end{theorem}

\begin{proof}
Assume that $c_2 \geq c_1 > 0$ and let the battlefields in $B$ be numbered in the order increasing with respect to their values, so that $v^1 \leq \ldots \leq v^m$.

As we already argued in the proof of Proposition~\ref{pr:maxeffort}, the upper bound on the equilibrium total effort of each player is attained when, for all $k \in B$, $p^k_1 = p^k_2 = 1/2$ and the upper bound on equilibrium total effort is 
\begin{equation*}
\sum_{k \in B} e_i^k = \frac{\gamma}{4c_i}\sum_{k \in B} v^k
\end{equation*}
We show that this bound is attained by the networks $\bm{\rho^{*}}_1$ and $\bm{\rho^{*}}_2$.

The adjacency matrix of $\bm{\rho^{*}}_i$ given by
\begin{equation*}
\bm{\rho^{*}}_i = \begin{bmatrix}
\bm{\sigma}^{1}_{i} \\
& \bm{\sigma}^{2}_{i} \\
& & \ddots \\
& & & \bm{\sigma}^{r}_{i}
\end{bmatrix}
\end{equation*}
where
\begin{equation*}
\bm{\sigma}^l_i = \begin{bmatrix}
0 & \lambda^{l}_{i} \\
\lambda^{l}_{i} & 0
\end{bmatrix}
\end{equation*}
if component $l$ consist of two nodes and, in the case of component $l$ consisting of three nodes,
\begin{equation*}
\bm{\sigma}^l_i = \begin{cases}
\begin{bmatrix}
0 & \frac{\lambda^{l}_{i}}{2} & \frac{\lambda^{l}_{i}}{2} \\
\frac{\lambda^{l}_{i}}{2} & 0 & \frac{\lambda^{l}_{i}}{2} \\
\frac{\lambda^{l}_{i}}{2} & \frac{\lambda^{l}_{i}}{2} & 0 \\
\end{bmatrix}, & \textrm{if $\frac{c_2}{c_1} < \frac{v^{m} + v^{m-1} + v^{m-2}}{v^{m} + v^{m-1} - v^{m-2}}$,}\\
\vspace{1mm}\\
\begin{bmatrix}
0 & \lambda^{l}_{i} & 0 \\
\lambda^{l}_{i} & 0 & 0 \\
\frac{\lambda^{l}_{i}}{2} & \frac{\lambda^{l}_{i}}{2} & 0 \\
\end{bmatrix}, & \textrm{if $\frac{c_2}{c_1} \geq \frac{v^{m} + v^{m-1} + v^{m-2}}{v^{m} + v^{m-1} - v^{m-2}}$.}
\end{cases}
\end{equation*}

Notice that
\begin{equation*}
\bm{\mu}_1 = \begin{bmatrix}
\frac{1}{1+\lambda^{1}_{1}} \\
\frac{1}{1+\lambda^{1}_{1}} \\
\frac{1}{1+\lambda^{2}_{1}} \\
\frac{1}{1+\lambda^{2}_{1}} \\
\vdots\\
\frac{1}{1+\lambda^{r}_{1}} \\
\frac{1}{1+\lambda^{r}_{1}} \\
\end{bmatrix} \quad \textrm{and} \quad \bm{\mu}_2 = \frac{c_1}{c_2} \bm{\mu}_1
\end{equation*}
solve equations $(\mathbf{I} + \bm{\rho^{*}}_i) \bm{\mu}_i = \bm{1}$, for $i \in \{1,2\}$, respectively.
In addition, using~\eqref{eq:char:pr} in Theorem~\ref{th:char}, for all $k \in B$,
\begin{equation*}
p_1^k = \frac{(\mu_{2}^k c_{2})^\gamma}{\left(\mu_{1}^k c_{1}\right)^{\gamma}+ \left(\mu_{2}^k c_{2}\right)^{\gamma}} = \frac{1}{2}.
\end{equation*}

To complete the proof for the upper bound on total efforts, we need to show that the game has an equilibrium with the values of $\bm{\mu}_1$ and $\bm{\mu}_2$ given above. Like in the proof of Proposition~\ref{pr:maxeffort}, this is done by applying~\eqref{eq:char:e} in Theorem~\ref{th:char}
and showing that the quantities for the efforts for both players, across all battlefields, are positive. 

Notice first that, for $i \in \{1,2\}$,
\begin{equation*}
\left(\mathbf{I}+\bm{\rho^{*}}_i\right)^{-1} = \begin{bmatrix}
\left(\mathbf{I} + \bm{\sigma}^{1}_{i}\right)^{-1} \\
& \left(\mathbf{I} + \bm{\sigma}^{2}_{i}\right)^{-1} \\
& & \ddots \\
& & & \left(\mathbf{I} + \bm{\sigma}^{r}_{i}\right)^{-1}
\end{bmatrix}
\end{equation*}
where
\begin{equation*}
\left(\mathbf{I}+\bm{\sigma}^l_i\right)^{-1} = \frac{1}{1+\lambda^{l}_{i}} \begin{bmatrix}
\frac{1}{1-\lambda^{l}_{i}} & \frac{\lambda^{l}_{i}}{\lambda^{l}_{i}-1} \\
\frac{\lambda^{l}_{i}}{\lambda^{l}_{i}-1} & \frac{1}{1-\lambda^{l}_{i}},
\end{bmatrix} 
= \mu^{l}_{i} \begin{bmatrix}
\frac{1}{1-\lambda^{l}_{i}} & \frac{\lambda^{l}_{i}}{\lambda^{l}_{i}-1} \\
\frac{\lambda^{l}_{i}}{\lambda^{l}_{i}-1} & \frac{1}{1-\lambda^{l}_{i}},
\end{bmatrix} 
\end{equation*}
for $l \in \{1,\ldots,r\}$ and the $l$'th component of size $2$, and if $l = r$ and the $l$'th component is of size $3$,
\begin{equation*}
\left(\mathbf{I}+\bm{\sigma}^l_i\right)^{-1} = \frac{1}{1+\lambda^{l}_{i}} \begin{bmatrix}
\frac{\lambda^{l}_i+2}{2-\lambda^{l}_{i}} & \frac{\lambda^{l}_i}{\lambda^{l}_{i}-2} & \frac{\lambda^{l}_i}{\lambda^{l}_{i}-2} \\
\frac{\lambda^{l}_i}{\lambda^{l}_{i}-2} & \frac{\lambda^{l}_i+2}{2-\lambda^{l}_{i}} & \frac{\lambda^{l}_i}{\lambda^{l}_{i}-2} \\
\frac{\lambda^{l}_i}{\lambda^{l}_{i}-2} & \frac{\lambda^{l}_i}{\lambda^{l}_{i}-2} & \frac{\lambda^{l}_i+2}{2-\lambda^{l}_{i}},
\end{bmatrix}
= \mu^{l}_{i} \begin{bmatrix}
\frac{\lambda^{l}_i+2}{2-\lambda^{l}_{i}} & \frac{\lambda^{l}_i}{\lambda^{l}_{i}-2} & \frac{\lambda^{l}_i}{\lambda^{l}_{i}-2} \\
\frac{\lambda^{l}_i}{\lambda^{l}_{i}-2} & \frac{\lambda^{l}_i+2}{2-\lambda^{l}_{i}} & \frac{\lambda^{l}_i}{\lambda^{l}_{i}-2} \\
\frac{\lambda^{l}_i}{\lambda^{l}_{i}-2} & \frac{\lambda^{l}_i}{\lambda^{l}_{i}-2} & \frac{\lambda^{l}_i+2}{2-\lambda^{l}_{i}},
\end{bmatrix}
\end{equation*}
if 
\begin{equation*}
\frac{c_2}{c_1} < \frac{v^{m} + v^{m-1} + v^{m-2}}{v^{m} + v^{m-1} - v^{m-2}},
\end{equation*}
and
\begin{equation*}
\left(\mathbf{I}+\bm{\sigma}^l_i\right)^{-1} = \frac{1}{1+\lambda^{l}_{i}} \begin{bmatrix}
\frac{1}{1-\lambda^{l}_{i}} & \frac{\lambda^{l}_i}{\lambda^{l}_{i}-1} & 0 \\
\frac{\lambda^{l}_i}{\lambda^{l}_{i}-1} & \frac{1}{1-\lambda^{l}_{i}} & 0 \\
-\frac{\lambda^{l}_i}{2} & -\frac{\lambda^{l}_i}{2} & 1+\lambda^{l}_{i},
\end{bmatrix}
= \mu^{l}_{i} \begin{bmatrix}
\frac{1}{1-\lambda^{l}_{i}} & \frac{\lambda^{l}_i}{\lambda^{l}_{i}-1} & 0 \\
\frac{\lambda^{l}_i}{\lambda^{l}_{i}-1} & \frac{1}{1-\lambda^{l}_{i}} & 0 \\
-\frac{\lambda^{l}_i}{2} & -\frac{\lambda^{l}_i}{2} & 1+\lambda^{l}_{i},
\end{bmatrix}
\end{equation*}
if 
\begin{equation*}
\frac{c_2}{c_1} \geq \frac{v^{m} + v^{m-1} + v^{m-2}}{v^{m} + v^{m-1} - v^{m-2}}.
\end{equation*}

Using~\eqref{eq:char:e} in Theorem~\ref{th:char},
\begin{align*}
\bm{e}_i & = \frac{\gamma}{c_i} \left(\mathbf{I} + {\bm{\rho^{*}}_i}^T\right)^{-1} \left(\bm{p}_1 \odot\bm{p}_2 \odot \bm{v} \oslash \bm{\mu}_i\right) \\
         & = \frac{\gamma}{4c_i} \left(\mathbf{I} + {\bm{\rho^{*}}_i}^T\right)^{-1}\left(\bm{v} \oslash \bm{\mu}_i\right)
\end{align*}
so that, in the case of $l \in \{1,\ldots,r\}$ such that the $l$'th group is of size $2$,
\begin{align*}
e^{2l-1}_1 & = \frac{\gamma}{c_1} \frac{\lambda^{l}_1 v^{2l} - v^{2l-1}}{4(\lambda^{l}_1 - 1)}\\
e^{2l}_1 & = \frac{\gamma}{c_1} \frac{\lambda^{l}_1 v^{2l-1} - v^{2l}}{4(\lambda^{l}_1 - 1)}\\
e^{2l-1}_2 & = \frac{\gamma}{c_2} \frac{c_1(v^{2l-1} + v^{2l}) - c_2 v^{2l}(\lambda^{l}_1 + 1)}{4(2c_1 - c_2(\lambda^{l}_1 + 1))}\\
e^{2l}_2 & = \frac{\gamma}{c_2} \frac{c_1(v^{2l-1} + v^{2l}) - c_2 v^{2l-1}(\lambda^{l}_1 + 1)}{4(2c_1 - c_2(\lambda^{l}_1 + 1))},
\end{align*}
and, in the case of $l = r$ and group $l$ of size $3$, 
if
\begin{equation*}
\frac{c_2}{c_1} < \frac{v^{m} + v^{m-1} + v^{m-2}}{v^{m} + v^{m-1} - v^{m-2}}
\end{equation*}
then
\begin{align*}
e^{m-2}_1 & = \frac{\gamma}{c_1} \frac{2v^1 - \lambda_1(v^3 + v^2 - v^1)}{4(2 - \lambda_1)}\\
e^{m-1}_1 & = \frac{\gamma}{c_1} \frac{2v^2 - \lambda_1(v^3 + v^1 - v^2)}{4(2 - \lambda_1)}\\
e^{m}_1 & = \frac{\gamma}{c_1} \frac{2v^3 - \lambda_1(v^2 + v^1 - v^3)}{4(2 - \lambda_1)}\\
e^{m-2}_2 & = \frac{\gamma}{c_2} \frac{c_1(v^{m} + v^{m-1} + v^{m-2}) - c_2 (v^{m} + v^{m-1} - v^{m-2})(\lambda^{r}_1 + 1)}{4(3c_1 - c_2(\lambda^{r}_1 + 1))}\\
e^{m-1}_2 & = \frac{\gamma}{c_2} \frac{c_1(v^{m} + v^{m-1} + v^{m-2}) - c_2 (v^{m} + v^{m-2} - v^{m-1})(\lambda^{r}_1 + 1)}{4(3c_1 - c_2(\lambda^{r}_1 + 1))}\\
e^{m}_2 & = \frac{\gamma}{c_2} \frac{c_1(v^{m} + v^{m-1} + v^{m-2}) - c_2 (v^{m-1} + v^{m-2} - v^{m})(\lambda^{r}_1 + 1)}{4(3c_1 - c_2(\lambda^{r}_1 + 1))},
\end{align*}
and if
\begin{equation*}
\frac{c_2}{c_1} \geq \frac{v^{m} + v^{m-1} + v^{m-2}}{v^{m} + v^{m-1} - v^{m-2}}
\end{equation*}
then
\begin{align*}
e^{m-2}_1 & = \frac{\gamma}{c_1} \frac{\lambda^{r}_1 v^{m-1} - v^{m-2}}{4(\lambda^{r}_1 - 1)}\\
e^{m-1}_1 & = \frac{\gamma}{c_1} \frac{\lambda^{r}_1 v^{m-2} - v^{m-1}}{4(\lambda^{r}_1 - 1)}\\
e^{m}_1 & = \frac{\gamma}{c_1} \frac{2v^{m} + \lambda^{r}_1 (2v^{m} - v^{m-1} - v^{m-2})}{8}\\
e^{m-2}_2 & = \frac{\gamma}{c_2} \frac{c_2 v^{m-1}(\lambda^{r}_1+1) - c_1(v^{m-2} + v^{m-1})}{4(c_2(\lambda^{r}_1 + 1) - 2c_1)}\\
e^{m-1}_2 & = \frac{\gamma}{c_2} \frac{c_2 v^{m-2}(\lambda^{r}_1+1) - c_1(v^{m-2} + v^{m-1})}{4(c_2(\lambda^{r}_1 + 1) - 2c_1)}\\
e^{m}_2 & = \frac{\gamma}{c_2} \frac{c_2 (2v^{m} - v^{m-1} - v^{m-2})(\lambda^{r}_1+1) + c_1(v^{m-2} + v^{m-1})}{8c_1}.
\end{align*}

By Lemmas~\ref{lemma:interior2}, \ref{lemma:interior3:1}, and~\ref{lemma:interior3:2}, below, all these quantities are greater than zero.
This shows that the game has an interior equilibrium $(\bm{e}_1,\bm{e}_2)$ with the associated marginal rates of substitution $(\bm{\mu}_1,\bm{\mu}_2)$.

Since the game on the networks $\bm{\rho^{*}}_1$ and $\bm{\rho^{*}}_2$, defined above, has an interior equilibrium $(\bm{e}_1,\bm{e}_2)$ with the associated marginal rates of substitution $(\bm{\mu}_1,\bm{\mu}_2)$, it attains the upper bound on equilibrium total efforts of the two players. 
By Lemma~\ref{lemma:unique} in the appendix and non-singularity of matrices $\mathbf{I} + \bm{\rho^{*}}_i$, for all $i \in \{1,2\}$, this constructed equilibrium is unique. This completes the proof.
\end{proof}

The lemmas below establish that the efforts of the players obtained in proof of Theorem~\ref{th:maxeffort} are positive.

\begin{lemma}
\label{lemma:interior2}
Suppose that $v^2 \geq v^1 > 0$ and $c_2 \geq c_1 > 0$.
If
\begin{equation}
\label{eq:interior2:1}
\frac{c_2}{c_1} < \frac{v^1 + v^2}{v^2} \textrm{ and } \lambda_1 \in \left[0,\frac{v^1 + v^2}{v^2}\frac{c_1}{c_2} - 1\right)
\end{equation}
or
\begin{equation}
\label{eq:interior2:2}
\frac{c_2}{c_1} \geq \frac{v^1 + v^2}{v^2} \textrm{ and } \lambda_1 > \frac{v^2}{v^1}
\end{equation}
then 
\begin{align*}
e_1^1 & = \frac{\gamma}{c_1} \frac{\lambda_1 v^2 - v^1}{4(\lambda_1 - 1)} > 0\\
e_1^2 & = \frac{\gamma}{c_1} \frac{\lambda_1 v^1 - v^2}{4(\lambda_1 - 1)} > 0\\
e_2^1 & = \frac{\gamma}{c_2} \frac{c_1(v^1 + v^2) - c_2 v^2(\lambda_1 + 1)}{4(2c_1 - c_2(\lambda_1 + 1))} > 0\\
e_2^2 & = \frac{\gamma}{c_2} \frac{c_1(v^1 + v^2) - c_2 v^1(\lambda_1 + 1)}{4(2c_1 - c_2(\lambda_1 + 1))} > 0.
\end{align*}
\end{lemma}

\begin{proof}
Since $v^2 \geq v^1 > 0$ so $e_1^k \geq 0$, for all $k \in \{1,2\}$, if either $\lambda_1 \in [0,v^1/v^2)$ (in which case $\lambda_1 < 1$) or $\lambda_1 > v^2/v^1$ (in which case $\lambda_1 > 1$).
Similarly, $e_2^k \geq 0$, for all $k \in \{1,2\}$, if either $\lambda_1 \in [0,(c_1/c_2)(v^1+v^2)/v^2 - 1)$ (in which case $\lambda_1 < 2c_1/c_2 - 1$) or $\lambda_1 > (c_1/c_2)(v^1+v^2)/v^1 - 1$ (in which case $\lambda_1 > 2c_1/c_2 - 1$).

If~\eqref{eq:interior2:1} is satisfied then the interval $[0,(c_1/c_2)(v^1+v^2)/v^2 - 1)$ is non-empty. Moreover, \eqref{eq:interior2:1}, together with $c_2 \geq c_1 > 0$, implies that $v^1/v^2 \geq (c_1/c_2)(v^1+v^2)/v^2 - 1$. Hence when~\eqref{eq:interior2:1} is satisfied, $e_i^k > 0$ for all $i \in \{1,2\}$ and $k \in \{1,2\}$.

If~\eqref{eq:interior2:2} is satisfied and $c_2 \geq c_1 > 0$ then $v^2/v^1 \geq (c_1/c_2)(v^1+v^2)/v^1 - 1$. Hence if~\eqref{eq:interior2:2} is satisfied then $e_i^k > 0$, for all $i \in \{1,2\}$ and $k \in \{1,2\}$.
\end{proof}

\begin{lemma}
\label{lemma:interior3:1}
Suppose that $v^3 \geq v^2 \geq v^1 > 0$ and $c_2 \geq c_1 > 0$.
If
\begin{equation}
\label{eq:interior3:1}
\frac{c_2}{c_1} < \frac{v^{3} + v^{2} + v^{1}}{v^{3} + v^{2} - v^{1}} \textrm{ and } \lambda_{1} \in \left[0,\frac{v^3 + v^2 + v^1}{v^3 + v^2 - v^1}\frac{c_1}{c_2} - 1\right)
\end{equation}
then
\begin{align*}
e^{1}_1 & = \frac{\gamma}{c_1} \frac{2v^1 - \lambda_1(v^3 + v^2 - v^1)}{4(2 - \lambda_1)} > 0\\
e^{2}_1 & = \frac{\gamma}{c_1} \frac{2v^2 - \lambda_1(v^3 + v^1 - v^2)}{4(2 - \lambda_1)} > 0\\
e^{3}_1 & = \frac{\gamma}{c_1} \frac{2v^3 - \lambda_1(v^2 + v^1 - v^3)}{4(2 - \lambda_1)} > 0\\
e^{1}_2 & = \frac{\gamma}{c_2} \frac{c_1(v^{3} + v^{2} + v^{1}) - c_2 (v^{3} + v^{2} - v^{1})(\lambda_1 + 1)}{4(3c_1 - c_2(\lambda_1 + 1))} > 0\\
e^{2}_2 & = \frac{\gamma}{c_2} \frac{c_1(v^{3} + v^{2} + v^{1}) - c_2 (v^{3} + v^{1} - v^{2})(\lambda_1 + 1)}{4(3c_1 - c_2(\lambda_1 + 1))} > 0\\
e^{3}_2 & = \frac{\gamma}{c_2} \frac{c_1(v^{3} + v^{2} + v^{1}) - c_2 (v^{2} + v^{1} - v^{3})(\lambda_1 + 1)}{4(3c_1 - c_2(\lambda_1 + 1))} > 0.
\end{align*}
\end{lemma}

\begin{proof}
The condition on the costs in~\eqref{eq:interior3:1} guarantees that the interval of possible values $\lambda_1$ is nonempty.
Since $c_2 \geq c_1 > 0$ and $v^3 \geq v^2 \geq v^1 > 0$ so
\begin{align*}
\frac{v^3 + v^2 + v^1}{v^3 + v^2 - v^1}\frac{c_1}{c_2} - 1 & \leq \frac{v^3 + v^2 + v^1}{v^3 + v^2 - v^1} - 1 = \frac{2v_1}{v^3 + v^2 - v^1} \leq 2.
\end{align*}
Hence $\lambda_1 < 2$ and the denominator in the quantities $e_1^k$ is positive, for all $k \in \{1,2,3\}$. Moreover, since $v^3 \geq v^2 \geq v_1$ so the numerator in $e_1^3$ is positive. In addition, 
\begin{equation*}
0 < \frac{2v_1}{v^3 + v^2 - v^1} \leq \frac{2v_2}{v^3 + v^1 - v^2}
\end{equation*}
and so the numerators in $e_1^1$ and $e_1^2$ are positive. Hence $e_1^k > 0$, for all $k \in \{1,2,3\}$.

Since $v^3 \geq v^2 \geq v_1$ so
\begin{equation*}
\frac{v^3 + v^2 + v^1}{v^3 + v^2 - v^1} < 3.
\end{equation*}
and 
\begin{equation*}
\frac{v^3 + v^2 + v^1}{v^3 + v^2 - v^1} \leq \frac{v^3 + v^2 + v^1}{v^3 + v^1 - v^2}.
\end{equation*}
This, together with the range of $\lambda_1$ assumed in~\eqref{eq:interior3:1}, implies that $\lambda_1 < 3c_1/c_2 - 1$. Hence the denominator in the quantities $e_2^k$ is positive, for all $k \in \{1,2,3\}$ and the numerator in the quantities $e_2^k$ is positive, for all $k \in \{1,2\}$.
In the case of $k = 3$, either $v_1 + v_2 - v_3 \leq 0$, in which case the numerator in $e_2^3 > 0$, or $v_1 + v_2 - v_3 > 0$, in which case
\begin{equation*}
\frac{v^3 + v^2 + v^1}{v^3 + v^2 - v^1} \leq \frac{v^3 + v^2 + v^1}{v^2 + v^1 - v^3}
\end{equation*}
and the numerator in $e_2^3 > 0$, because $\lambda_1$ satisfies~\eqref{eq:interior3:1}. 
Hence $e_2^k > 0$, for all $k \in \{1,2,3\}$.
\end{proof}

\begin{lemma}
\label{lemma:interior3:2}
Suppose that $v^3 \geq v^2 \geq v^1 > 0$ and $c_2 \geq c_1 > 0$.
If
\begin{equation}
\label{eq:interior3:2}
\frac{c_2}{c_1} \geq \frac{v^{3} + v^{2} + v^{1}}{v^{3} + v^{2} - v^{1}} \textrm{ and } \lambda_{1} > \frac{v^2}{v_1}
\end{equation}
then
\begin{align*}
e^{1}_1 & = \frac{\gamma}{c_1} \frac{\lambda_1 v^{2} - v^{1}}{4(\lambda_1 - 1)} > 0\\
e^{2}_1 & = \frac{\gamma}{c_1} \frac{\lambda_1 v^{1} - v^{2}}{4(\lambda_1 - 1)} > 0\\
e^{3}_1 & = \frac{\gamma}{c_1} \frac{2v^{3} + \lambda_1 (2v^{3} - v^{2} - v^{1})}{8} > 0\\
e^{1}_2 & = \frac{\gamma}{c_2} \frac{c_2 v^{2}(\lambda_1+1) - c_1(v^{1} + v^{2})}{4(c_2(\lambda_1 + 1) - 2c_1)} > 0\\
e^{2}_2 & = \frac{\gamma}{c_2} \frac{c_2 v^{1}(\lambda_1+1) - c_1(v^{1} + v^{2})}{4(c_2(\lambda_1 + 1) - 2c_1)} > 0\\
e^{3}_2 & = \frac{\gamma}{c_2} \frac{c_2 (2v^{3} - v^{2} - v^{1})(\lambda_1+1) + c_1(v^{1} + v^{2})}{8c_1} > 0.
\end{align*}
\end{lemma}

\begin{proof}
Since $c_2 \geq c_1 > 0$ and $v_3 \geq v_2 \geq v_1 > 0$ so $e^3_i >0$, for all $i \in \{1,2\}$, and
\begin{equation*}
\frac{v^2}{v^1} \geq 1 \textrm{ and } \frac{v^2}{v^1} \geq \frac{c_1}{c_2}\frac{v^1 + v^2}{v^1} - 1 \geq \frac{c_1}{c_2}2 - 1.
\end{equation*}
This implies that $e^k_1 > 0$ and $e^k_2$, for all $k \in \{1,2\}$.
\end{proof}

\subsection{Multiple equilibria}

In this section we provide an example showing that possibility of multiple equilbria in the model, even if matrices $\mathbf{I} + \bm{\rho}_i$ and all their minors are non-singular, for all $i \in \{1,2\}$. 

Consider a scenario with $B = \{1,2,3\}$. Suppose that the network of spillovers of player $1$, $\bm{\rho}_1$, is as presented in Figure~\ref{fig:multeq} and network of spillovers of player $2$ is empty.

\begin{figure}[htp]
\begin{center}
  \includegraphics[scale=0.7]{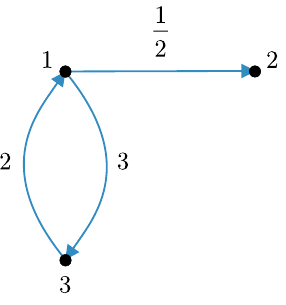}
\end{center}
  \caption{Network of player $1$.}
  \label{fig:multeq}
\end{figure}

 The adjacency matrices of spillovers for the two players are
\begin{equation*}
\bm{\rho}_{1} = \begin{bmatrix}
            0 & \frac{1}{2} & 3 \\
            0 & 0 & 0 \\
            2 & 0 & 0
            \end{bmatrix} \qquad\qquad\qquad
\bm{\rho}_{2} = \begin{bmatrix}
            0 & 0 & 0 \\
            0 & 0 & 0 \\
            0 & 0 & 0
            \end{bmatrix}.
\end{equation*}

Notice that matrix $\mathbf{I} + \bm{\rho}_1$ as well as all its minors are non-singular.
The CSF is Tullock CSF with $\gamma = 1$, the values of all battlefields are equal to $1$. The costs are equal to $1$, $c_1 = c_2 = 1$.

Let $P_1 = \{1,2,3\}$ and $P_2 = \{1,2\}$. Solving the system of equations defined by~\eqref{eq:char:mu} and~\eqref{eq:char:pr} in~Theorem~\ref{th:char} we compute the quantities $\bm{\mu}_1$ and $\bm{\mu}_2$:
\begin{align*}
\mu_1^1 & = \frac{1}{2} \textrm{, } & \mu_1^2 & = 1 \textrm{, } 
& \mu_1^3 & = 0 \\
\mu_2^1 & = 1 \textrm{, } & \qquad \mu_2^2 & = 1, & &
\end{align*}
and $\mu_2^3$ takes any non-negative real value.
Using~\eqref{eq:char:pr} we obtain the winning probabilities:
\begin{equation*}
p_1^1 = \frac{2}{3}\textrm{, } \qquad p_1^2 = \frac{1}{2},\qquad
p_1^3 = 1.
\end{equation*}

Inserting the $\bm{\mu}$'s and the probabilities into~\eqref{eq:char:e} we obtain a system of equations for finding the efforts. Solving it we obtain
\begin{align*}
e_1^1 & = \frac{18-4 \mu^3_2}{45 \mu^3_2} \textrm{, } & e_1^2 & = \frac{53 \mu^3_2 - 36 }{180 \mu^3_2} \textrm{, } & e_1^3 & = \frac{4 \mu^3_2-3}{15 \mu^3_2} \\
e_2^1 & = \frac{2}{9} \textrm{, } & e_2^2 & = \frac{1}{4} \textrm{, } & e_2^3 & = 0.
\end{align*}
This strategy profile is valid if $3/4 < \mu^3_2 < 9/2$. Notice, in particular, that each of the values of $\mu^2_3$ determines a different equilibrium and there is continuum of Nash equilibria. The equilibrium effective efforts are
\begin{align*}
y_1^1 & = \frac{4}{9} \textrm{, } & y_1^2 & = \frac{1}{4} \textrm{, } & y_1^3 & = \frac{1}{\mu^3_2} \\
y_2^1 & = \frac{2}{9} \textrm{, } & y_2^2 & = \frac{1}{4} \textrm{, } & y_2^3 & = 0.
\end{align*}
In any of these equilibria, battlefield $3$ receives $0$ effective effort from player $2$ and each battlefield receives positive effective effort from player $1$.

The equilibrium total efforts of the two players are equal to
\begin{equation*}
\sum_{k \in B} e^k_1 = \sum_{k \in B} e^k_2 = \frac{17}{36}
\end{equation*}
and the equilibrium payoffs are
\begin{equation*}
\payoff{1}(\bm{e}_1,\bm{e}_2) = \frac{61}{36}, \qquad\qquad \payoff{2}(\bm{e}_1,\bm{e}_2) = \frac{13}{36}.
\end{equation*}
 Notice, in particular, that equilibrium probabilities of winning battlefields, equilibrium total efforts, and equilibrium payoffs are the same across all equilibria.

\end{document}